%% file: main.tex
\documentclass[preprintnumbers, amsmath, amssymb, superscriptaddress, nofootinbib]{revtex4-1}

\usepackage[english]{babel}
\usepackage[T1]{fontenc}
\usepackage[utf8]{inputenc}
\usepackage{amsthm}
\usepackage{graphicx}
\usepackage{float}
\usepackage{subcaption}
\usepackage[ruled,vlined]{algorithm2e}
\usepackage{enumitem}
\usepackage{tikz}
\usepackage{physics}
\usepackage{dsfont}
\usepackage{hyperref}
\usepackage{cleveref}
\crefname{app}{Appendix}{Appendices}
\crefname{table}{box}{boxes}
\Crefname{Table}{Box}{Boxes}

\newcommand{\mc}[1]{\mathcal{#1}}
\newcommand{\tl}[1]{\tilde{#1}}
\newtheorem{proposition}{Proposition}

\newcommand{\boxalign}[2][0.97\textwidth]{
  \par\noindent\tikzstyle{mybox} = [draw=black,inner sep=6pt]
  \begin{center}\begin{tikzpicture}
   \node [mybox] (box){%
    \begin{minipage}{#1}{\vspace{-5mm}#2}\end{minipage}
   };
  \end{tikzpicture}\end{center}
}

\begin{document}

\title{Designing Quantum Networks Using Preexisting Infrastructure}
\date{\today}

\author{Julian Rabbie}
\email{julianrabbie@gmail.com}
\affiliation{QuTech, Delft University of Technology, Lorentzweg 1, 2628 CJ Delft, The Netherlands}
\affiliation{Kavli Institute of Nanoscience, Delft University of Technology, Lorentzweg 1, 2628 CJ Delft, The Netherlands}
\author{Kaushik Chakraborty}
\thanks{These authors contributed equally.}
\affiliation{QuTech, Delft University of Technology, Lorentzweg 1, 2628 CJ Delft, The Netherlands}
\affiliation{Kavli Institute of Nanoscience, Delft University of Technology, Lorentzweg 1, 2628 CJ Delft, The Netherlands}
\author{Guus Avis}
\thanks{These authors contributed equally.}
\affiliation{QuTech, Delft University of Technology, Lorentzweg 1, 2628 CJ Delft, The Netherlands}
\affiliation{Kavli Institute of Nanoscience, Delft University of Technology, Lorentzweg 1, 2628 CJ Delft, The Netherlands}
\author{Stephanie Wehner}
\email{s.d.c.wehner@tudelft.nl}
\affiliation{QuTech, Delft University of Technology, Lorentzweg 1, 2628 CJ Delft, The Netherlands}
\affiliation{Kavli Institute of Nanoscience, Delft University of Technology, Lorentzweg 1, 2628 CJ Delft, The Netherlands}

\input{abstract.tex}

\maketitle

\input{introduction.tex}
\input{results.tex}
\input{discussion.tex}
\input{methods.tex}
\input{acknowledgements.tex}
\input{contributions}
\bibliographystyle{vancouver}
\bibliography{bibliography}
\clearpage
\appendix
\input{appendix.tex}
\end{document}

%% file: abstract.tex
\begin{abstract}
We consider the problem of deploying a quantum network on an existing fiber infrastructure, where quantum repeaters and end nodes can only be housed at specific locations. We propose a method based on integer linear programming (ILP) to place the minimal number of repeaters on such an existing network topology, such that requirements on end-to-end entanglement-generation rate and fidelity between any pair of end-nodes are satisfied. While ILPs are generally difficult to solve, we show that our method performs well in practice for networks of up to 100 nodes. We illustrate the behavior of our method both on randomly-generated network topologies, as well as on a real-world fiber topology deployed in the Netherlands.
\end{abstract}

%% file: introduction.tex
\section{Introduction}
\label{sec:introduction}

The quantum internet will provide an infrastructure for quantum communication between any two devices in the world  \cite{Van14, LSWK04, Kim08, WEH18}.
This can be used to perform tasks which are provably impossible with the classical internet.
Many of these are cryptographic in nature and allow unconditional security, such as quantum key distribution \cite{bb14,E91}, secure multi-party cryptography \cite{BC16} and blind quantum computation \cite{broadbentUniversalBlindQuantum2009}.
Other applications of the quantum internet include fast byzantine agreement \cite{ben-orFastQuantumByzantine2005} and clock synchronization \cite{komarQuantumNetworkClocks2014a}. \\

A major challenge in the construction of terrestrial quantum networks is to overcome exponential loss in optical fibers.
In order to enable quantum communication over large distances, quantum repeaters are required.
These can form a quantum-repeater chain in which consecutive nodes are connected by elementary links.
Quantum repeaters are a very active research area and major advances have been achieved recently \cite{bhaskarExperimentalDemonstrationMemoryenhanced2020, stephensonHighrateHighfidelityEntanglement2019, rozpedekNeartermQuantumrepeaterExperiments2019, humphreysDeterministicDeliveryRemote2018, yuEntanglementTwoQuantum2020}.
However, the technology is not yet at the stage of practical deployment, 
and we anticipate that the first practical quantum repeaters will be costly.
It seems likely that before a global quantum internet is effected, smaller quantum networks connecting a limited number of end nodes are deployed.
A cost-efficient way of deploying such networks is using existing classical infrastructure by converting already-deployed optical fiber and installing quantum repeaters at strategic locations.\\

We model a classical fiber network which forms the basis of a quantum network as an undirected, weighted graph $G = (\mathcal N, \mathcal F, \mathcal L)$.
The nodes $\mathcal N$ are partitioned into a set of \emph{end nodes} $\mathcal C \subset \mathcal N$ and a set of \emph{potential repeater locations} $\mathcal R = \mathcal N \setminus \mathcal C$.
The goal of the quantum network is to enable quantum communication between end nodes.
Potential repeater locations are any location in the network where a quantum repeater could be placed.
Such a location could, for example, be a hub in the classical network with the facilities required to run a quantum repeater.
The edges of the graph are the fibers of the network, $\mathcal F$, where $\mathcal L (f)$ is the length of fiber $f \in \mathcal F$.
In case a quantum repeater is installed at a potential repeater location, the potential repeater location becomes a \emph{quantum-repeater node}.
When deploying a quantum network based on a classical fiber network, it is essential to determine which potential repeater locations should be turned into quantum-repeater nodes.
\\

In order to have an operational quantum network, nodes must be connected by elementary links.
For many quantum-repeater schemes (such as those using heralded entanglement generation \cite{inside_quantum_repeaters}), elementary links consist of fibers with active elements measuring qubits.
Therefore, when deploying a quantum network based on a classical fiber network, it must also be determined which fibers to convert into elementary links.
Here, we consider that elementary links can be constructed from any number of consecutively-adjacent fibers in the graph $G$ (passing through potential repeater locations).
Both fibers and potential repeater locations can be part of multiple elementary links, which is motivated by the fact that fibers are typically constructed in bundles (meaning that each elementary link could, in fact, use the same fiber bundle but a different fiber). Additionally, multiplexing over different wavelengths could be used to enable the use of a single fiber in multiple elementary links.
For an example of how a (very) small classical fiber network can be used to create a quantum network, see \Cref{fig:four_node_ill}. \\

\begin{figure}[b]
    \captionsetup{justification=raggedright}
    \includegraphics[width=0.35\textwidth]{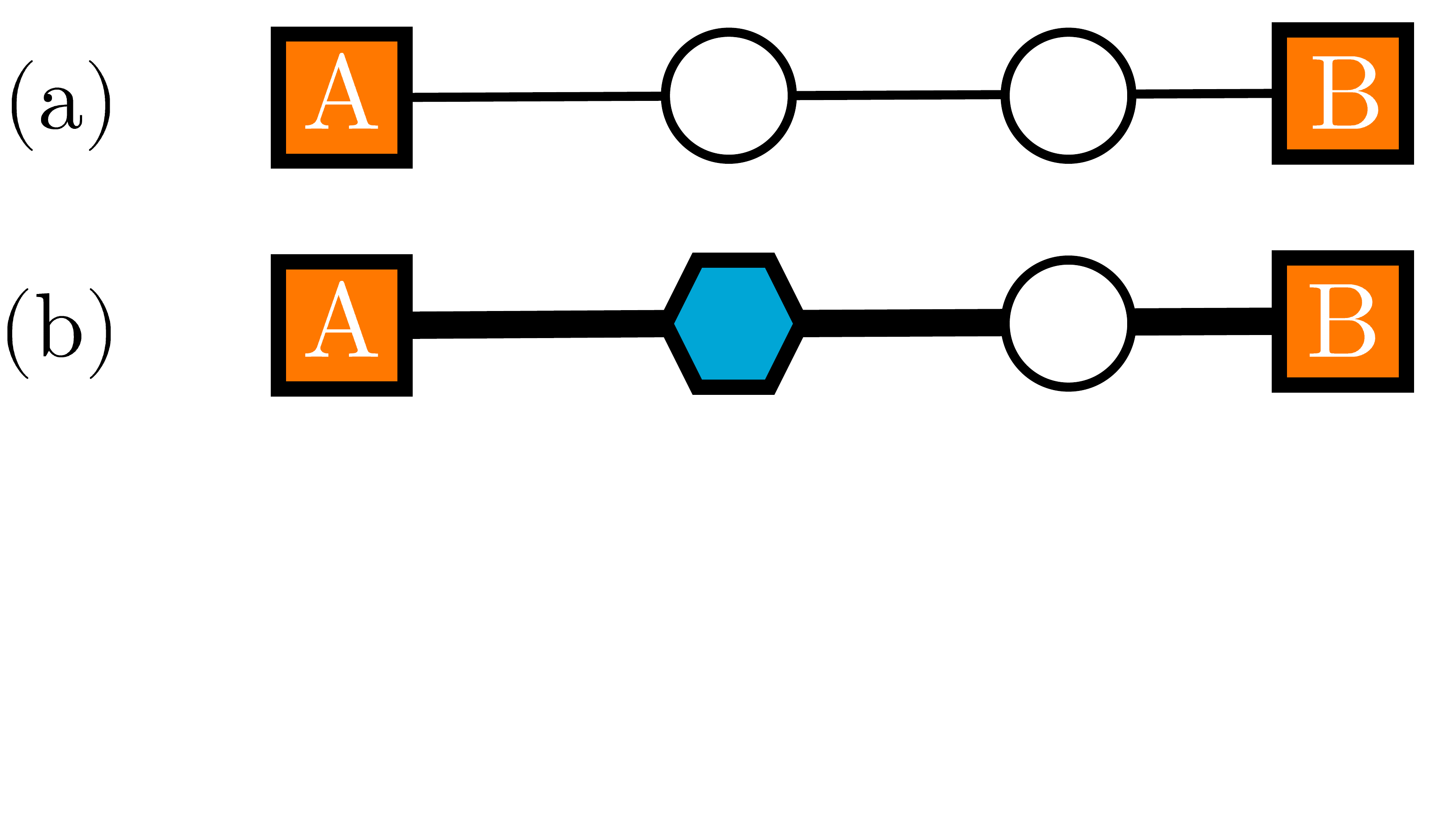}
    \vspace{-35pt}
    \caption{Example of how a quantum network can be constructed. \textbf{(a)} Graph representing a simple fiber network. Nodes A and B are end nodes, while the other two nodes are potential repeater locations. \textbf{(b)} Quantum network that is constructed using the preexisting fiber network. The first node from A is used as a quantum-repeater node (blue hexagon) and there are two elementary links. One elementary link is made from A to the quantum-repeater node, while the other starts at the quantum-repeater node and ends at B.}
    \label{fig:four_node_ill}
\end{figure}

Here, we introduce the problem of determining how to construct a quantum network using a preexisting classical fiber network as the repeater-allocation problem.
We define it as follows:

\pagebreak 

\noindent \textbf{Repeater-Allocation Problem}:
Given a classical fiber network corresponding to the undirected, weighted graph $G = (\mathcal N, \mathcal F, \mathcal L)$ with end nodes $\mathcal C \subset \mathcal N$.
Which of the potential repeater locations $\mathcal R = \mathcal N \setminus \mathcal C$ should be turned into quantum-repeater nodes,
and which fibers should be converted into elementary links,
such that a quantum network is obtained which satisfies a set of network requirements, while the associated costs are minimized? \\

In this paper
we present, to the best of our knowledge for the first time,
a method which solves the repeater-allocation problem.
Here, we only consider the costs associated to installing quantum repeaters, as we expect that the first practical quantum repeaters will come at a high cost.
Furthermore, the set of network requirements that we consider are the following:
\begin{enumerate}[label=\textbf{\arabic*}, ref={\arabic*}]
	\item \emph{Rate and fidelity}. \\
	The quantum network must be able to distribute bipartite entangled quantum states between any pair of end nodes at some minimum rate, which we denote $R_\text{min}$.
	Furthermore, the states must have some minimum fidelity to a maximally entangled state, which we denote $F_\text{min}$.
	The network must be able to do this for every pair of end nodes simultaneously. 
	
	In a quantum-repeater chain with fixed hardware, the rate of entanglement distribution is limited by loss and noise in elementary links and in quantum repeaters. 
	Therefore, it is generally possible to lower bound the rate by upper bounding the number of quantum repeaters (and thereby the number of elementary links), and the length of each elementary link (assuming the photon loss probability per unit length is constant).
	Similarly, fidelity is limited by noisy operations in quantum repeaters, while it can also be a decreasing function of the elementary link length (this can be, for example, due to dark counts in detectors).
	Therefore, fidelity too can be lower bounded by upper bounding the number of quantum repeaters and the elementary link length.
	
	We use these bounds to assess whether the rate and fidelity between a pair of end nodes is sufficient.
	For any $R_\text{min}$ and $F_\text{min}$, we can find $N_\text{max}$ and $L_\text{max}$ such that a repeater chain of $N_\text{max}$ repeaters and elementary links of length $L_\text{max}$ can deliver entangled states at rate $R_\text{min}$ with fidelity $F_\text{min}$.
	Then, we consider two end nodes capable of receiving entangled states with at least rate $R_\text{min}$ and at least fidelity $F_\text{min}$ 
	if there is a free path between them which contains at most  $N_\text{max}$ repeaters and of which each elementary link is at most $L_\text{max}$ long.
	
	How exactly $N_\text{max}$ and $L_\text{max}$ can be determined from $R_\text{min}$ and $F_\text{min}$ is specific to the quantum-repeater architecture and depends on various performance parameters.
	We give a toy-model calculation in \Cref{subsec:methods:toy-model} as an example.
	Note that when considering a quantum-repeater architecture which is not based on entanglement distribution, the method presented in this paper is still applicable if a performance metric like rate and fidelity can be determined which can be lower bounded by upper bounding the number of repeaters and the elementary link lengths of a repeater chain.
	
	\label{requirement:quality}
	
	\item \emph{Robustness}. \\
	When a part of a quantum network breaks down, all other requirements should still be met.
	We quantify this using the minimum number of quantum-repeater nodes or elementary links (it can be any combination) that need to break down before one of the other requirements can no longer be met.
	Here, we use the symbol $K$ to refer to this number. 
	\label{requirement:robustness}
	
	\item \emph{Repeater capacity}.\\
	Quantum-repeater nodes should never be required to operate above their capacity in order to meet all other network requirements.
	We define the capacity of a quantum-repeater node as the maximum number of quantum-communication sessions it can facilitate simultaneously.
	In an entanglement-based network, this could be directly related to the number of entangled states that can be stored in memory or the number of Bell-state measurements that can be performed simultaneously.
	Here, we use the symbol $D$ to refer to the capacity of the quantum-repeater nodes.
	\label{requirement:repeater_capacity}
\end{enumerate}

%% file: results.tex
\section{Results}
\label{sec:results}

In this section we present a method, detailed in \Cref{method}, which aids in the design of a quantum network using existing classical infrastructure.
Specifically, given a fiber network, our method makes it possible to choose at which locations quantum repeaters should be installed.
This is done such that entangled states can be distributed between all pairs of end nodes simultaneously with a minimum rate and fidelity.
Furthermore, our method guarantees that the resulting quantum network is robust against failure of quantum repeaters and elementary links, and can take finite capacity of quantum repeaters into account.
At the same time, our method minimizes the total number of quantum repeaters that need to be installed. 
We dub the problem that our method solves the repeater-allocation problem. 
\smallskip \\ 

\begin{table}[ht]
	\boxalign{
		\vspace{.5cm}
		\textbf{Input}
		\begin{itemize}
			\item Fiber network graph $G = (\mathcal N, \mathcal F, \mathcal L)$.
			\item Set of end nodes $\mathcal C \subset \mathcal N$. 
			\item Minimum rate $R_\text{min}$ and fidelity $F_\text{min}$ required by end nodes.\footnotemark[1]
			\item Required robustness parameter $K$ (number of quantum-repeater nodes and elementary links that must be incapacitated before network operation is compromised).
			\item Capacity parameter $D$ (number of quantum-communication sessions that one quantum repeater can facilitate simultaneously).
    		\end{itemize}
    		\vspace{.5cm}
    		\textbf{Method}
    		\begin{enumerate}
    			\item Determine values for the parameters $L_\text{max}$ and $N_\text{max}$ such that a quantum-repeater chain consisting of $N_\text{max}$ repeaters and elementary links of length $L_\text{max}$ is able to deliver entangled states at rate $R_\text{min}$ with fidelity $F_\text{min}$ to a maximally entangled state.
    			\item Construct the set of potential repeater locations
    			\begin{equation} \label{eq:R}
    			\mathcal R = \mathcal N \setminus \mathcal C.
    			\end{equation}
    			\item Construct the set
    			\begin{equation}
    			\mathcal Q = \Big \{ (s, t)| (s, t) \in_R \{ (i, j), (j, i) \}, i, j \in C, i \neq j \Big \}, \label{eq:Q}
    			\end{equation}
    			
    			where $\in_R$ implies picked uniformly at random.
			\item For every $(s, t) \in \mathcal Q$, construct the set
			\begin{equation}\label{eq:E}
			\mathcal E _{(s, t)} = \Big \{ (n_1, n_2) | n_1 \in \mathcal R \cup \{s\}, n_2 \in \mathcal R \cup \{t\} , n_1 \neq n_2\Big \},
			\end{equation}
			and then construct the set
			\begin{equation}
			 \mathcal E = \bigcup_{q \in \mathcal Q} \mathcal E_q.
			\end{equation}
			\item For every $(u,v) \in \mathcal{E}$, determine the shortest path from $u$ to $v$ in the fiber-network graph $G$.
			Store the length of the path as $L\Big( (u,v) \Big)$ and the fibers making it up as $F\Big( (u,v) \Big)$.
			\item Solve the link-based formulation in \Cref{lbf} using an ILP solver. 
			Store the values of the variables $x^{q,k}_{uv}$ and $y_u$.
			\item Apply the path extraction algorithm, i.e. \Cref{alg:path_extraction}, to obtain the set $\mc P^*$.
			For every $(u,v) \in \mathcal E$, set $x_{uv}^{q,k} = 0$ if there is no $p \in \mathcal P^*$ such that $(u,v) \in p$.
			\label{step:remove_loops}
		\end{enumerate}
		\vspace{.5cm}
		\textbf{Solution}
		\begin{itemize}
			\item Every potential repeater location $u \in \mathcal R$ for which $y_u = 1$ should be used as a quantum-repeater node.
			\item For every $(u, v) \in \mathcal E$ for which $x_{uv}^{q,k} = 1$ for some value of $q$ and $k$, an elementary link should be constructed using the fibers $F\Big( (u,v) \Big)$.
		\end{itemize}
	}
	\caption{Method to solve the repeater-allocation problem.}
	\label{method}
\end{table}

Key to our method is integer linear programming (ILP), which can be used to obtain the optimal repeater placement with an optimization solver such as Clp \cite{coinor}, Gurobi \cite{gurobi} or CPLEX \cite{ibm2019cplex}. Our method has been tested both using a real fiber network and a large number of randomized graphs, on which we report in \Cref{subsec:discussion:real_network,subsec:discussion:evaluation_on_random_networks} respectively. 
The real network contains four end nodes and 50 potential repeater locations, and a solution was found in 74 seconds using a computer running a 
quad-core Intel Xeon W-2123 processor at 3.60 GHz and 16 GB of RAM, demonstrating that the method is feasible for realistically-sized networks. \\ 

\renewcommand{\thefootnote}{\alph{footnote}}
\footnotetext[1]{Instead of a minimum rate and fidelity, one can also use the minimum value(s) for other performance metric(s), as long as these can be lower bounded by upper bounding the number of repeaters and elementary link lengths of a quantum-repeater chain.}

Here, we put forward two different ILP formulations.
The first, which we call the \emph{path-based formulation} (see \Cref{pbf}), is based on enumerating and then choosing 
paths between end nodes of the quantum network.
It is relatively easy to show and understand that this formulation indeed solves the repeater-allocation problem (see \Cref{subsec:methods:validity_of_pbf}). 
However, it is not efficient, as the number of variables and constraints in the formulation 
grows exponentially with the size of the network.
The second formulation is the \emph{link-based formulation} (see \Cref{lbf}).
This formulation is much more efficient than the path-based formulation, as it only grows polynomially with the size of the network.
Therefore, our method as described in \Cref{method} uses the link-based formulation.
It is, however, harder to see that the link-based formulation can be used to solve the repeater-allocation problem.
Yet, the link-based formulation is equivalent to the path-based formulation, as we show in Section \ref{subsec:methods:proof_of_equivalence}. \\

The structure of the paper is as follows.
In the remainder of this section, we present
our method for solving the repeater-allocation problem and introduce both the intuitive path-based formulation and the efficient link-based formulation.
Next, in \Cref{sec:discussion}, we first give an example of the use of our method on a real fiber network in the Netherlands.
We also study the behaviour and performance of the method on a large number of randomly-generated network graphs.
Furthermore, we present ways in which our method can be extended, and we discuss its limitations.
Finally, in \Cref{sec:methods}, we argue that the path-based formulation can indeed be used to solve the repeater-allocation problem, we give an example of a rate-fidelity analysis, we sketch a proof of
the equivalence of the path-based formulation and the link-based formulation,
we explain how we generate random network graphs and we present the scaling of the two ILP formulations. 

\subsection{Path-Based Formulation}
\label{subsec:results:pbf}

The main idea behind the path-based formulation, which is shown in \Cref{pbf}, is to enumerate and then choose paths
for every $(s, t) \in \mathcal Q$, where $\mathcal Q$ is the set of all ordered pairs of end nodes as defined in \Cref{eq:Q}.
A path between $s$ and $t$ is a sequence of elementary links that does not contain any loops and connects $s$ and $t$. 
Quantum-repeater nodes are then allocated in such a way that they enable the chosen paths to be used. This can be considered an instance of the set cover problem \cite{daskin2011network}.
To guarantee a minimum rate $R_\text{min}$ and fidelity $F_\text{min}$, we require every chosen path to contain at most $N_\text{max}$ quantum-repeater nodes,
and we require every elementary link in the path to be at most $L_\text{max}$ long.
$N_\text{max}$ and $L_\text{max}$ are functions of $R_\text{min}$ and $F_\text{min}$, and what these functions look like depends on the specific quantum-repeater implementation under consideration.
For an example of how $N_\text{max}$ and $L_\text{max}$ can be derived from $R_\text{min}$ and $F_\text{min}$, see \Cref{subsec:methods:toy-model}.
Furthermore, to guarantee the network is robust, we choose $K$ different paths per end-node pair.
They are chosen such that none of the $K$ paths share a quantum-repeater node or an elementary link.
Finally, to account for the finite capacity of quantum repeaters, we choose the paths such that every quantum-repeater node is only used by at most $D$ different paths.
It can be intuitively understood that any quantum network accommodating the use of all these paths, will satisfy all network requirements considered in this paper. \\

\begin{table}[h]
\boxalign{
\begin{alignat}{2}
    \hspace*{-100pt}
    \min \quad \sum_{u \in \mathcal R} y_u & \phantom{\leq} \label{eq:pbf_objfun} \\
    \text{s.t. } \quad 
    L\Big((u, v)\Big)x_p &\leq L_\text{max} \qquad && \forall (u, v) \in p, p \in \mathcal P \label{eq:pbf_max_length} \\
    |p| x_p &\leq N_\text{max} + 1 \qquad && \forall p \in \mathcal P \label{eq:pbf_max_repeaters}\\
    \sum_{p \in \mathcal{P}_q} x_p & = K && \forall q \in \mathcal Q  \label{eq:pbf_K}\\
    \sum_{p \in \mathcal{P}_q}r_{up}x_p & \leq 1 && \forall u \in \mathcal{R}, q \in \mathcal{Q} \label{eq:pbf_disjoint}\\
    \sum_{p \in \mathcal P}r_{up}x_p & \leq Dy_u && \forall u \in \mathcal{R} \label{eq:pbf_capacity}  \\
     x_p & \in \{0, 1\} && \forall p \in \mathcal P \label{eq:pbf_x}\\
     y_{u} & \in \{0,1\} && \forall u \in \mathcal R \label{eq:pbf_y} \\
     \mathrm{where} \quad
     r_{up} &= 
     \begin{cases}
     1 \qquad \text{if path $p$ uses $u$ as a quantum-repeater node} \\
     0 \qquad \text{otherwise}
     \end{cases}
     \qquad
     &&\forall u \in \mathcal R, p \in \mathcal P 
\end{alignat}
}
\caption{Path-based formulation.}
\label{pbf}
\end{table}

Key to the path-based formulation are the binary decision variables $x_p$, which are defined for every path $p \in \mathcal P = \cup_{(s, t) \in \mathcal Q} \mathcal P_{(s, t)}$, where $\mathcal P_{(s, t)}$ is the set of all possible paths from end node $s$ to end node $t$.
The elementary links that can be contained by a path $p \in \mathcal P_{(s, t)}$ must all be in $\mathcal E_{(s, t)}$, which is defined in Equation \eqref{eq:E}.
Each $x_p$ has value $1$ when $p$ is considered part of the chosen set of paths, and $0$ otherwise. 
Furthermore, there are the binary decision variables $y_u$ for all $u \in \mathcal R$. 
$y_u$ is $1$ if a quantum repeater is placed at potential repeater location $u$, and $0$ otherwise.
Constraints \eqref{eq:pbf_max_length} to \eqref{eq:pbf_capacity} guarantee that these variables are chosen such that all network requirements are satisfied.
The objective function \eqref{eq:pbf_objfun} ensures that they are chosen such that the total number of quantum-repeater nodes is minimized.
It is argued that solutions to the path-based formulation are indeed solutions to the repeater-allocation problem in \Cref{subsec:methods:validity_of_pbf}.\\

The path-based formulation requires us to define one variable $x_p$ corresponding to each path $p \in \mc P$. 
Hence, the total number of variables as well as the number of constraints are at least $|\mc P|$, which is $O(|\mathcal N|!)$. 
Therefore the size of the input to the ILP solver scales exponentially with the number of nodes. 
This makes the path-based formulation unsuitable for designing quantum networks based on large fiber networks.
Our implementation of the path-based formulation in CPLEX can be found in the repository \cite{githubcode}.
In the next section, we give a more efficient formulation. 

\subsection{Link-Based Formulation}
\label{subsec:results:lbf}

Here we present the link-based formulation, which can be found in \Cref{lbf}. 
This formulation is inspired by the capacitated facility location problem \cite{daskin2011network}.
Instead of choosing which paths to use, we choose which elementary links to use.
Quantum repeaters can then be placed such that each chosen elementary link is enabled.
To this end, for each end-node pair $q \in \mc Q$, for every elementary link $(u,v) \in \mathcal E_q$ and for $k = 1, 2, \ldots, K$,  we define the binary decision variable $x^{q,k}_{uv}$. 
It can be thought of as indicating whether elementary link $(u,v)$ is used in the $k^\text{th}$ path used to connect end node $s$ to end node $t$, where $q = (s, t)$.
Furthermore, we again use the variables $y_u$ that indicate whether node $u \in \mathcal R$ is used as a quantum-repeater node. \\

\begin{table}[h]
\boxalign{
\begin{alignat}{2}
    \hspace*{-100pt}
    \min \quad \sum_{u \in \mathcal{R}} y_u & \phantom{\leq} \label{eq:lbf_objfun} \\
    \mathrm{s.t.} \quad 
    \sum_{\substack{v \\ (u, v) \in \mathcal E _q}} x^{q,k}_{uv} - \sum_{\substack{v \\ (v, u) \in \mathcal E _q}}x^{q,k}_{vu} & =
    \begin{cases}
    1, \quad & \text{if } u = s \\
    -1, \quad & \text{if } u = t \\
    0, \quad & \text{if } u \in \mathcal{R}
    \end{cases}
    \qquad
    &&\forall u \in \mathcal R \cup \{s, t\}, q = (s, t) \in \mathcal{Q}, k = 1, 2, \ldots, K \label{eq:lbf_flowcon} \\
    L\Big( (u, v) \Big) x^{q,k}_{uv} & \leq L_{\mathrm{max}} &&  \forall (u,v) \in \mathcal E_q,  q \in \mathcal{Q}, k = 1, 2, \ldots, K \label{eq:lbf_max_length} \\
    \sum_{(u,v) \in \mathcal E _q}  x^{q,k}_{uv} & \leq N_{\mathrm{max}} +1 && \forall q \in \mathcal{Q}, k = 1, 2, \ldots, K \label{eq:lbf_max_repeaters} \\
    \sum_{\substack{v \\ (u, v) \in \mathcal E_q}}\sum_{k=1}^K x_{uv}^{q,k} & \leq 1 && \forall u \in \mathcal{R}, q \in \mathcal{Q} \label{eq:lbf_disjoint_repeaters} \\
    \sum_{k=1}^K x_{st}^{q,k} & \leq 1 && \forall q \in \mathcal{Q} \label{eq:lbf_max_one_direct_path} \\
    \sum_{q \in\mathcal{Q}} \sum_{\substack{v \\ (u, v) \in \mathcal E _q}} \sum_{k = 1}^{K} x^{q,k}_{uv} & \leq D y_u && \forall u \in \mathcal{R} \label{eq:lbf_capacity} \\
    x^{q,k}_{uv} & \in \{0, 1\} && \forall (u,v) \in \mathcal{E}_q, q \in \mathcal{Q}, k = 1, 2, \ldots, K \label{eq:xdmon_lbf_k} \\
    y_u & \in \{0, 1\} && \forall u \in \mathcal{R} \label{eq:ydom_lbf}
\end{alignat}
}
\caption{Link-based formulation.}
\label{lbf}
\end{table}

Because the number of elementary links scales polynomially with the number of nodes, both the number of variables and the number of constraints also scale polynomially with the number of nodes $|\mc N|$.
In particular, they are $O (|\mc N|^2)$ (see \Cref{subsec:methods:scaling} for a derivation).
Our implementation of the link-based formulation in CPLEX can be found in the repository \cite{githubcode}. \\

In \Cref{subsec:methods:proof_of_equivalence}, we sketch the proof of the equivalence of the path-based formulation and the link-based formulation.
Furthermore, we sketch why the variables $x_{uv}^{q,k}$ and $y_u$ still provide a solution to the link-based formulation after performing step \ref{step:remove_loops} of \Cref{method}.
The reason this step is included in our method is because, otherwise, elementary links could be included in the solution which are not necessary to meet the network requirements.
The detailed version of the proof can be found in \Cref{app:proof_of_equivalence}.
Since the link-based formulation scales much more favourably with the size of the fiber network under consideration, 
it is more efficient to use this formulation when solving the repeater-allocation problem for large networks.

%% file: discussion.tex
\section{Discussion}
\label{sec:discussion}

In this section we illustrate our method as implemented by the link-based formulation using the Python API of CPLEX version 12.9 \cite{ibm2019cplex}. The corresponding code can be found in the repository \cite{githubcode}.
Furthermore, we investigate the effect of varying network-requirement parameters and discuss possible extensions and limitations of our method.

\subsection{Example on a Real Network}
\label{subsec:discussion:real_network}

Here, we demonstrate our method by solving the repeater-allocation problem for a real fiber network.
The fiber network that we consider is the core network of SURFnet.
The latter is a network provider for Dutch educational and research institutions and has provided us with the network data, which is available in the repository \cite{githubcode}.
The network graph is depicted in \Cref{fig:surfnet_data}. \\

\begin{figure}[t]
	\centering
	\captionsetup{justification=raggedright}
	\includegraphics[width=1\textwidth]{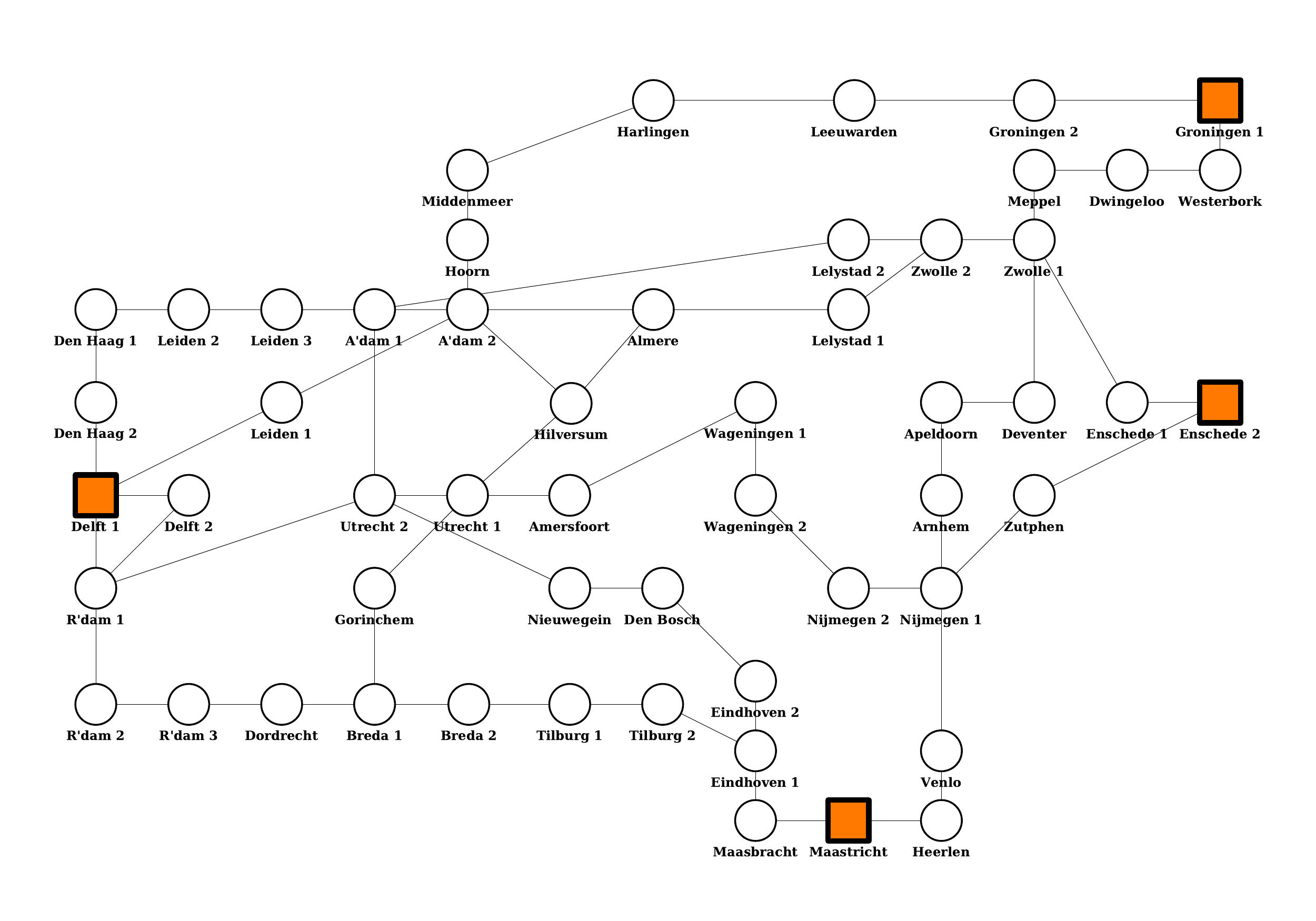}
	\caption{Graph representation of SURFnet core network. 
		Node locations roughly correspond to geographical locations but have been adjusted for readability.
		Lengths of fibers connecting nodes are not shown.
		Nodes that are used as end nodes are shown as orange squares.
		Potential repeater locations are shown as white circles. 
		A'dam and R'dam are used as abbreviations for Amsterdam and Rotterdam respectively.}
	\label{fig:surfnet_data}
\end{figure}

As end nodes of the network, we have chosen the cities of Delft, Enschede, Groningen and Maastricht.
In this example, we consider an entanglement-based quantum network utilizing massive multiplexing as described in e.g. \cite{sinclairSpectralMultiplexingScalable2014}.
For the end nodes, we require a minimum rate of $R_\text{min} = 1$ Hz (one entangled state per second) and a fidelity to a maximally entangled state $F_\text{min} = 0.93$.
Furthermore, we set the robustness parameter to $K=2$ (thus requiring that any single quantum repeater or elementary link in the network can break down without compromising network functionality), and we set the capacity parameter to $D=4$ (which, in this case, means that we assume each quantum repeater can perform four Bell-state measurements simultaneously).\\

The first step of our method requires us to calculate the $L_\text{max}$ and $N_\text{max}$ corresponding to the minimal rate and fidelity we have chosen.
This requires us to study the behaviour of a quantum-repeater chain consisting of $N + 1$ elementary links of length $L$ each.
$L_\text{max}$ and $N_\text{max}$ then have to be chosen as the largest possible values for $L$ and $N$ respectively such that the repeater chain still achieves the required rate and fidelity.
Here, we make a couple of simplifying assumptions to make the calculations more tractable.
Particularly, we assume elementary links generate Werner states, and we assume that the only losses are due to fiber attenuation and probabilistic Bell-state measurements (which we take to have a 50\% success probability).
In \Cref{subsec:methods:toy-model}, we perform the calculation and find that for an elementary-link fidelity $F_\text{link} = 0.99$, number of multiplexing modes $M=1000$, speed of light in fiber $c_\text{fiber} = 200,000$ km/s and attenuation length $L_\text{att} = 22$ km, we have $N_\text{max} = 6$ and $L_\text{max} = 136$ km. \\

The rest of the steps of the method in \Cref{method} have been performed using a Python script and CPLEX \cite{githubcode}.
The resulting solution is shown graphically in \Cref{fig:surfnet_solution}.
All chosen repeater nodes are shown as blue hexagons, while all fibers that are used in elementary links are drawn as thick lines. We see that repeaters are placed around Groningen in order to bridge the large distance to the other end nodes without exceeding the maximum elementary-link length $L_\text{max}$. Additionally, placing quantum-repeater nodes close to Groningen means they can be used for several of Groningen's outgoing connections.
There are multiple such nodes close together because each only has a limited capacity ($D = 4$), and the redundancy increases the robustness of the network.\\

\begin{figure}
    \captionsetup{justification=raggedright}
	\centering
	\includegraphics[width=.9\textwidth]{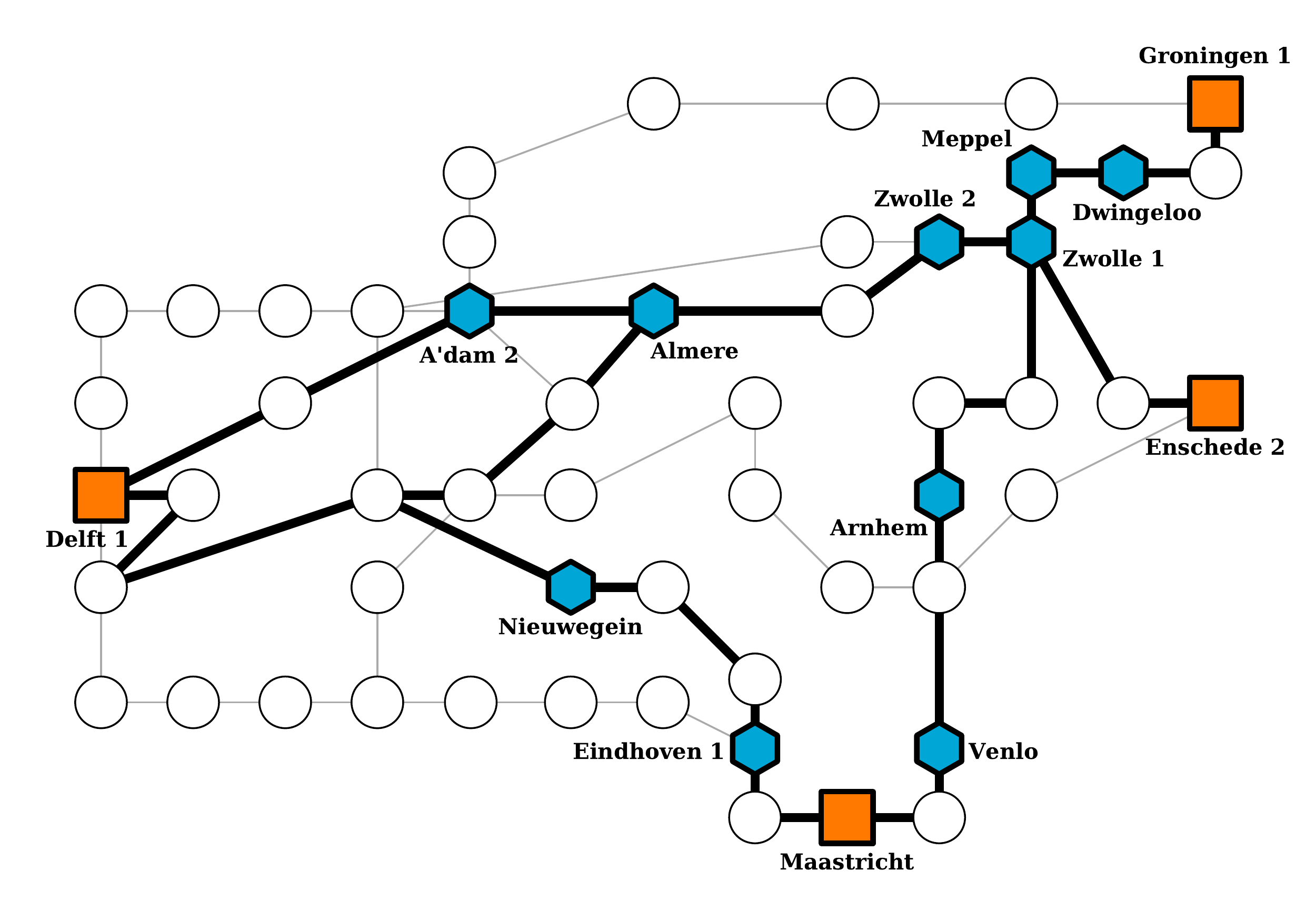}
	\caption{Solution to the repeater-allocation problem for $R_\text{min} = 1$ Hz, $F_\text{min} = 0.93$, $K = 2$ and $D = 4$. 
	The network graph used as input corresponds to the SURFnet network, depicted in Figure \ref{fig:surfnet_data}.
	End nodes are shown as orange squares, quantum-repeater nodes are shown as blue hexagons and the fibers that are used in the elementary links are highlighted with thick lines.
	}
	\label{fig:surfnet_solution}
\end{figure}

On our setup (see \Cref{sec:results}), it took us approximately 74 seconds to find the optimal solution to the link-based formulation for this network. Note that a feasible solution is a combination of decision variable values that satisfy all the constraints, while the optimal solution is a feasible solution that also minimizes the objective function.

\subsection{Effect of Network-Requirement Parameters} \label{subsec:discussion:evaluation_on_random_networks}

Here, we demonstrate and investigate the effect of the different network-requirement parameters on the outcome of our method.
The network-requirement parameters are, in principle, the minimum rate $R_\text{min}$, the minimum fidelity $F_\text{min}$, the robustness parameter $K$ and the capacity parameter $D$.
However, since $R_\text{min}$ and $F_\text{min}$ are translated into a maximum number of repeaters $N_\text{max}$ and a maximum elementary-link length $L_\text{max}$ in our method,  we here consider the network-requirement parameters to be $L_\text{max}$, $N_\text{max}$, $K$ and $D$.
This way, we can keep our discussion agnostic about the exact hardware used to create a quantum network and how $R_\text{min}$ and $F_\text{min}$ are mapped to $N_\text{max}$ and $L_\text{max}$. 
\\

First we give a visual demonstration on how the different network-requirement parameters affect the repeater placement.
To this end, we have created a network graph with end nodes in the corners of the network and 10 possible repeater locations randomly distributed in between the end nodes. 
For details on how the graph was obtained, see \Cref{subsec:methods:generating_random_networks}.
While keeping the network fixed, we vary the network-requirement parameters $D$, $K$ and $L_\text{max}$.
In \Cref{fig:square_k1,fig:square_k2,fig:square_k3}, we explore how the robustness parameter influences the total number of required quantum repeaters. 
Since each repeater has a capacity of $D = 6$ to distribute entanglement between the six end-node pairs,
and because the network is set up in such a way that each path needs exactly one quantum-repeater node to connect end nodes without elementary links exceeding $L_\text{max}=0.9$,
the optimal solution always contains $K$ repeaters. 
In \Cref{fig:square_d1,fig:square_d2,fig:square_d3} on the other hand, we see that as the capacity of quantum repeaters is varied from $D=1$ to $D=3$, the required number of quantum repeaters decreases when $D$ increases. Note that since $K = 1$, the optimal solution here always happens to contain $|\mathcal{Q}|/D$ repeaters.
Finally, in \Cref{fig:square_Lmax06,fig:square_Lmax075,fig:square_Lmax09} we see that as we allow for longer elementary links to be used, the total number of repeaters is decreased. If we would increase $L_\text{max}$ even further, at a certain point every end node can be connected to another end node with a direct elementary link and hence the number of repeaters will drop to zero.
The degeneracy of the optimal solution is visible from the fact that the solutions with two repeaters for $K = 2$ (\Cref{fig:square_k2}), $D = 3$ (\Cref{fig:square_d3}) and $L_\text{max} = 0.75$ (\Cref{fig:square_Lmax075}) are not equal. 
In \Cref{subsec:discussion:extensions}, it is discussed how this degeneracy can be lifted.
We do not show the effect of $N_\text{max}$.
Since the total number of repeaters is already minimized, changing the value of $N_\text{max}$ does not change the repeater allocation, but only determines whether a feasible solution exists at all. \\

Considering how the repeater placement on a single network varies with the network-requirement parameters can offer insight into how our method operates. However, it does not provide a general investigation into the properties of our method. 
In order to make more general and quantitative statements about our method, we will next consider the effect of varying network-requirement parameters on the repeater allocation for an ensemble of random networks.
In this work, we construct random network graphs using random geometric graphs.
That is, network graphs are constructed by scattering nodes randomly on a unit square.
Edges are put only between nodes if the Euclidean distance separating them is smaller than some number, which is called the radius of the random geometric graph.
The nodes which form the convex hull of the network are chosen as end nodes, so that the others are potential repeater locations. This choice is motivated by the fact that any potential-repeater locations that do not lie between end nodes would probably not play an important role anyway.
For a more elaborate account of how we generate random network graphs, see \Cref{subsec:methods:generating_random_networks}. \\

\begin{figure}[H]
    \centering
    \begin{subfigure}[b]{0.3\textwidth}
		\includegraphics[width=\textwidth]{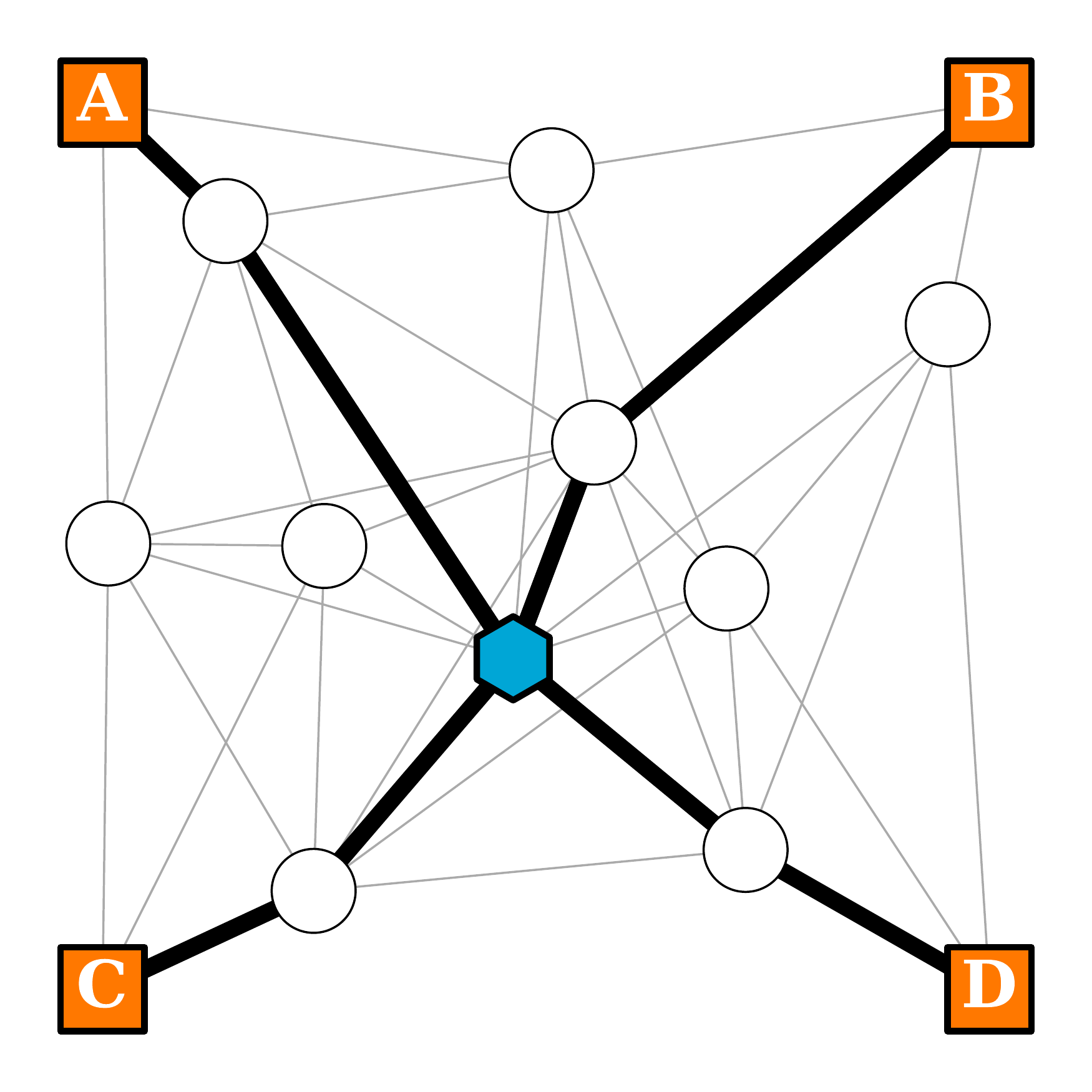}
		\caption{$K=1$}
		\label{fig:square_k1}
	\end{subfigure}
	~
	\begin{subfigure}[b]{0.3\textwidth}
		\includegraphics[width=\textwidth]{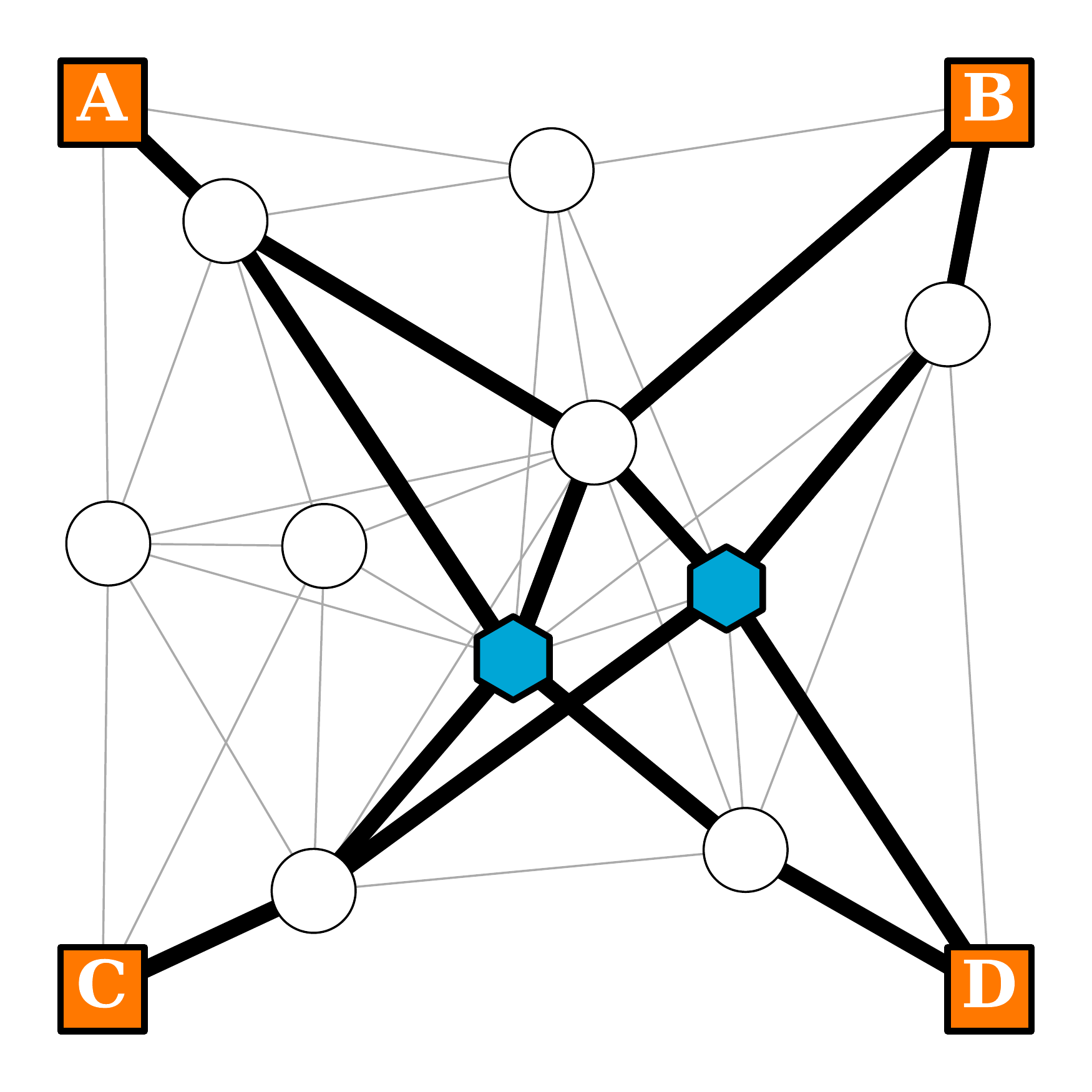}
		\caption{$K=2$}
		\label{fig:square_k2}
	\end{subfigure}
	~
	\begin{subfigure}[b]{0.3\textwidth}
		\includegraphics[width=\textwidth]{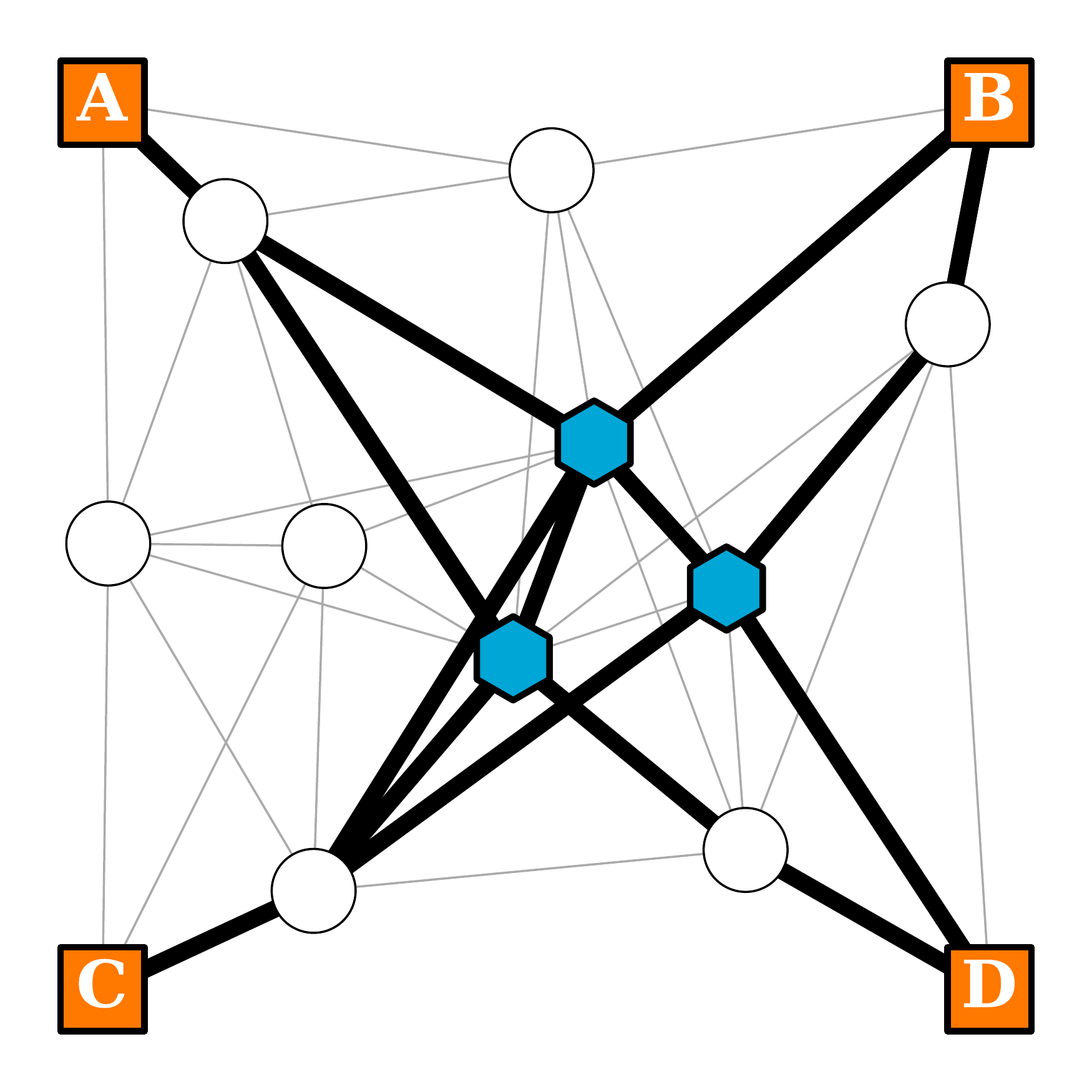}
		\caption{$K=3$}
		\label{fig:square_k3}
	\end{subfigure}
	\begin{subfigure}[b]{0.3\textwidth}
		\includegraphics[width=\textwidth]{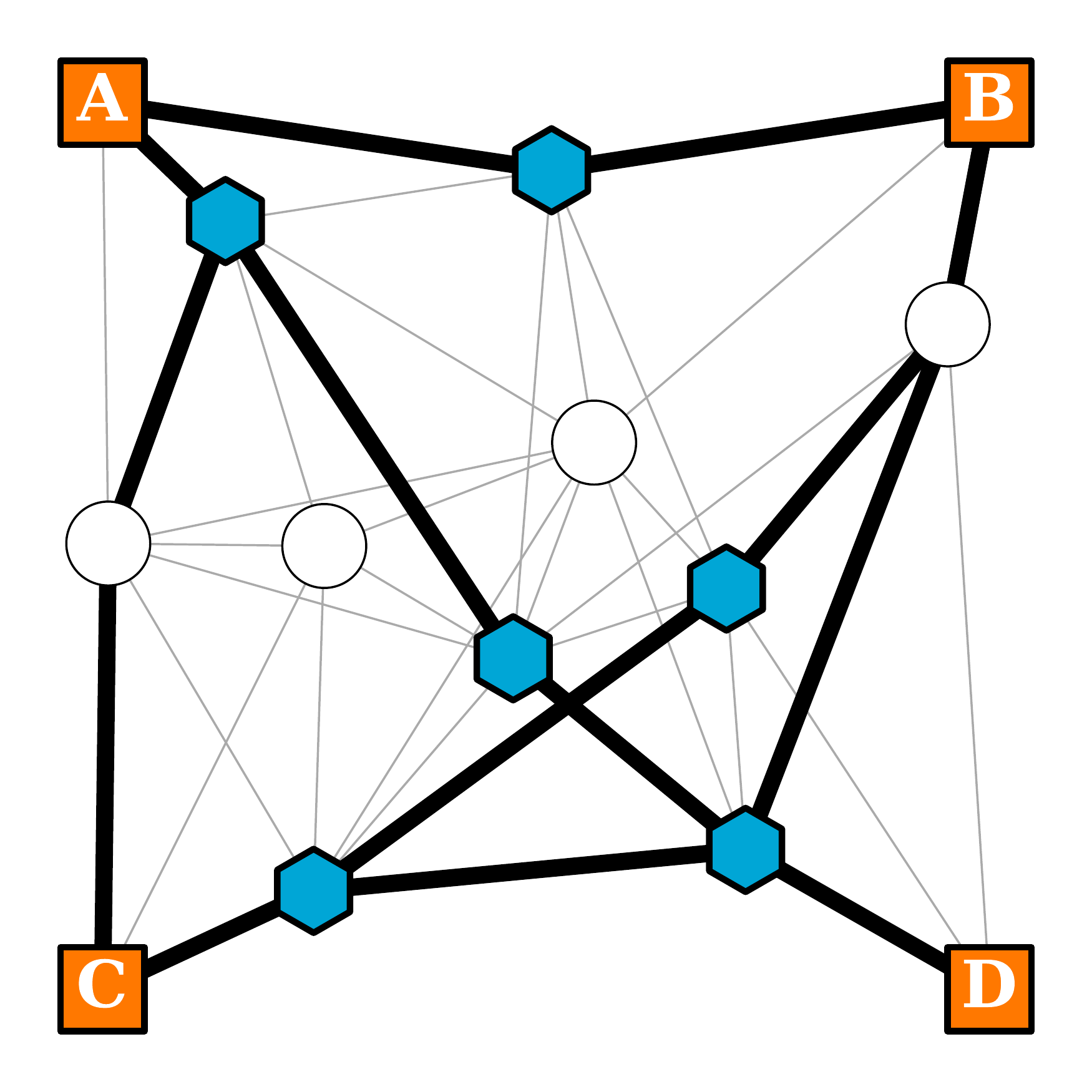}
		\caption{$D=1$}
		\label{fig:square_d1}
	\end{subfigure}
	~
	\begin{subfigure}[b]{0.3\textwidth}
		\includegraphics[width=\textwidth]{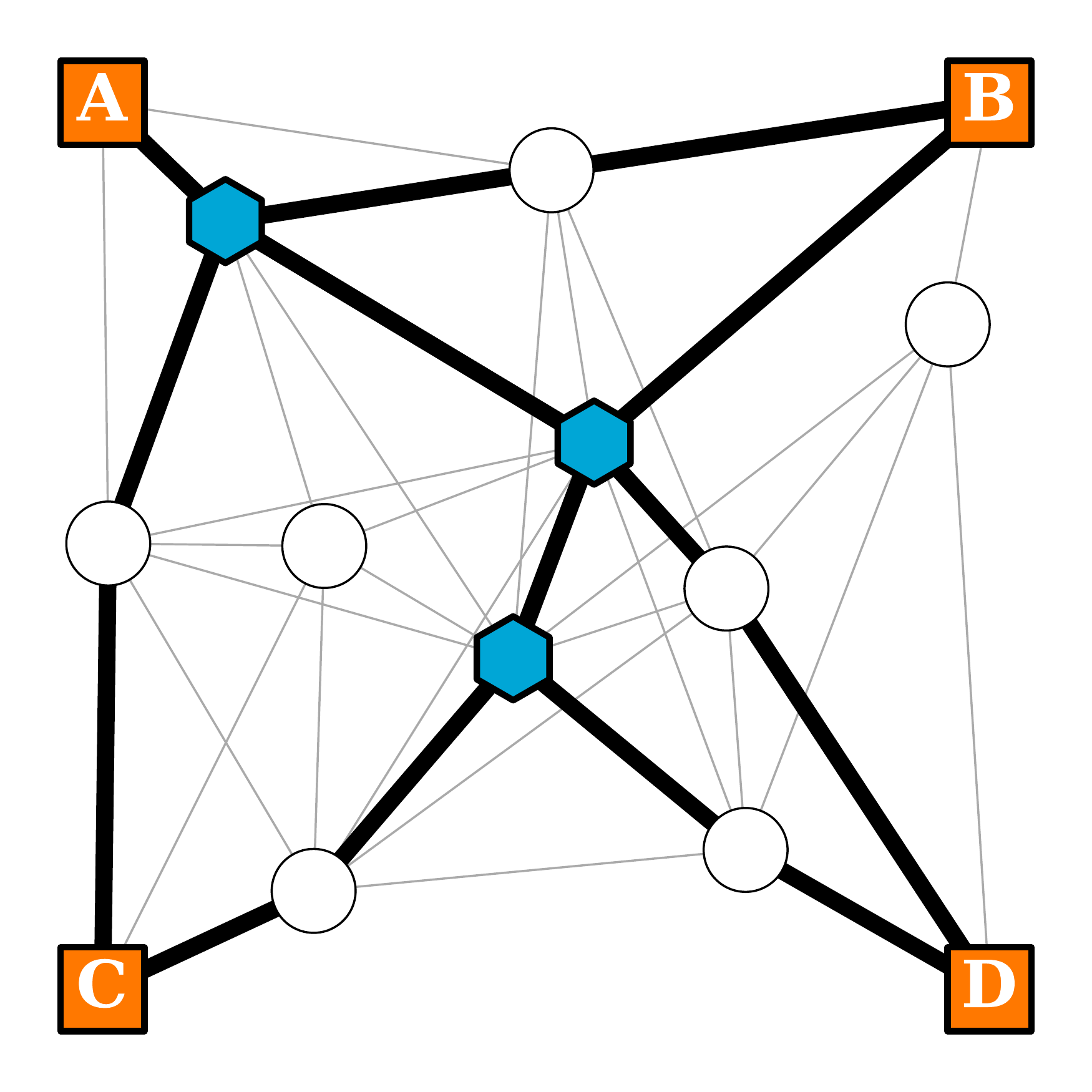}
		\caption{$D=2$}
		\label{fig:square_d2}
	\end{subfigure}
	~
	\begin{subfigure}[b]{0.3\textwidth}
		\includegraphics[width=\textwidth]{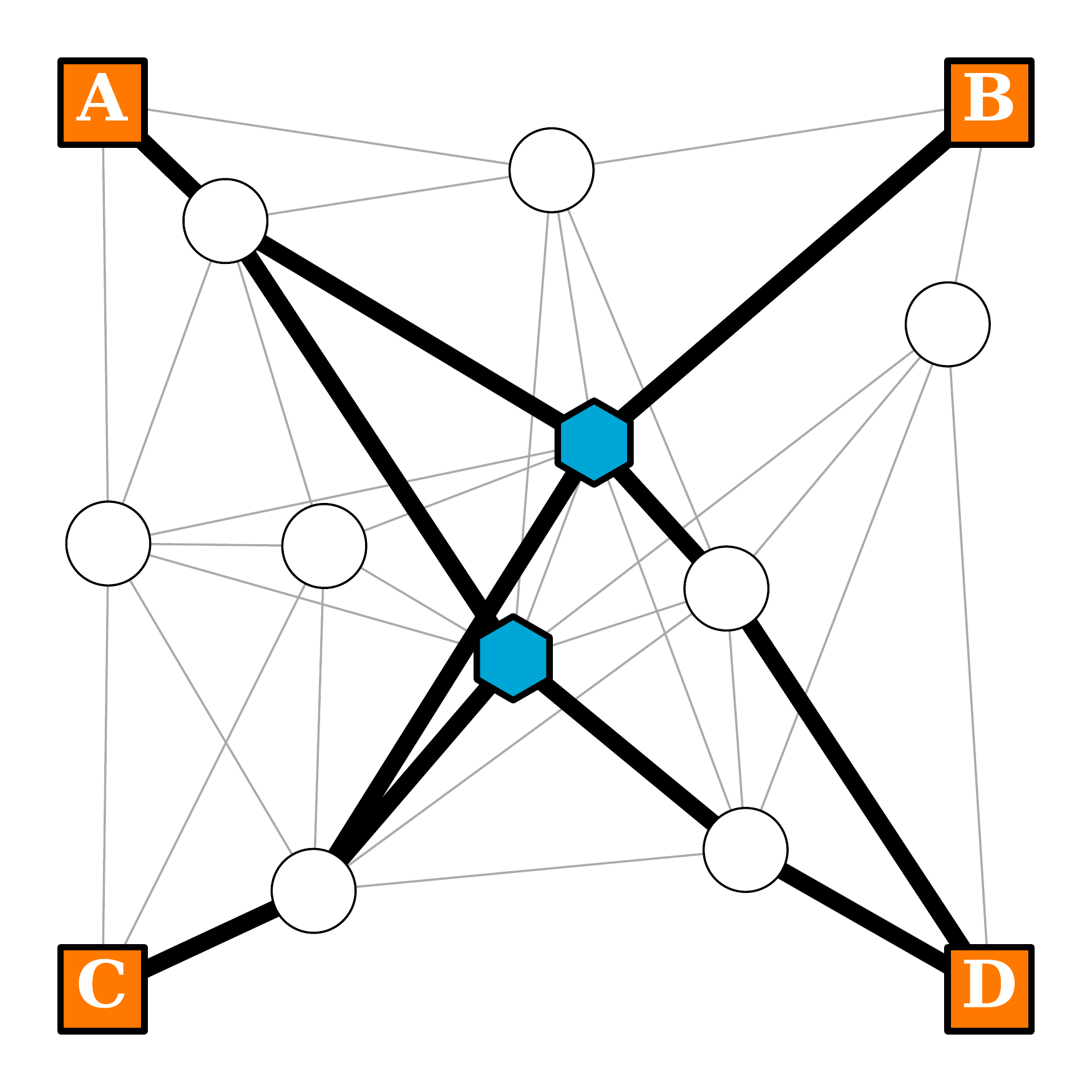}
		\caption{$D=3$}
		\label{fig:square_d3}
	\end{subfigure}
	\begin{subfigure}[b]{0.3\textwidth}
		\includegraphics[width=\textwidth]{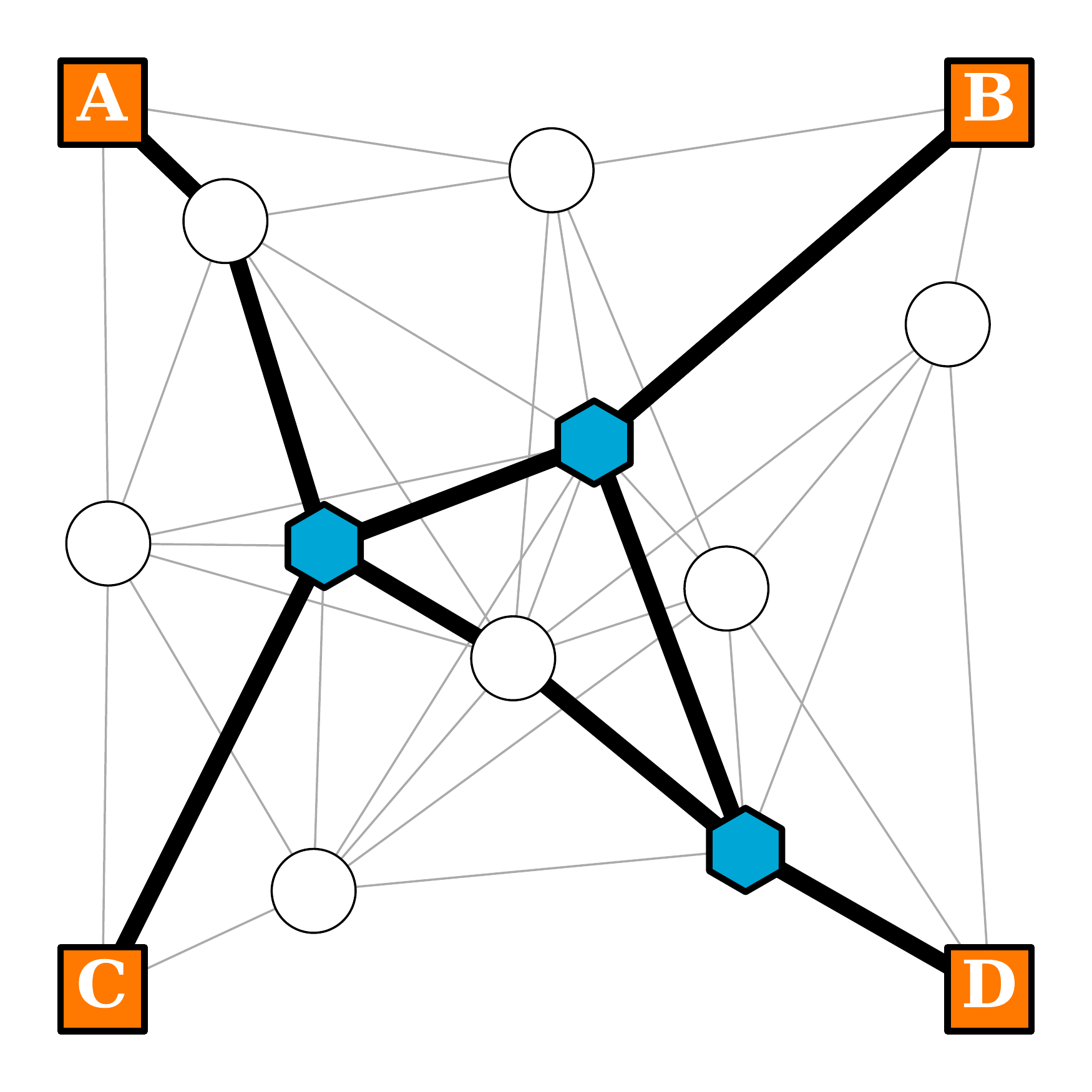}
		\caption{$L_\text{max} = 0.6$}
		\label{fig:square_Lmax06}
	\end{subfigure}
	~
	\begin{subfigure}[b]{0.3\textwidth}
		\includegraphics[width=\textwidth]{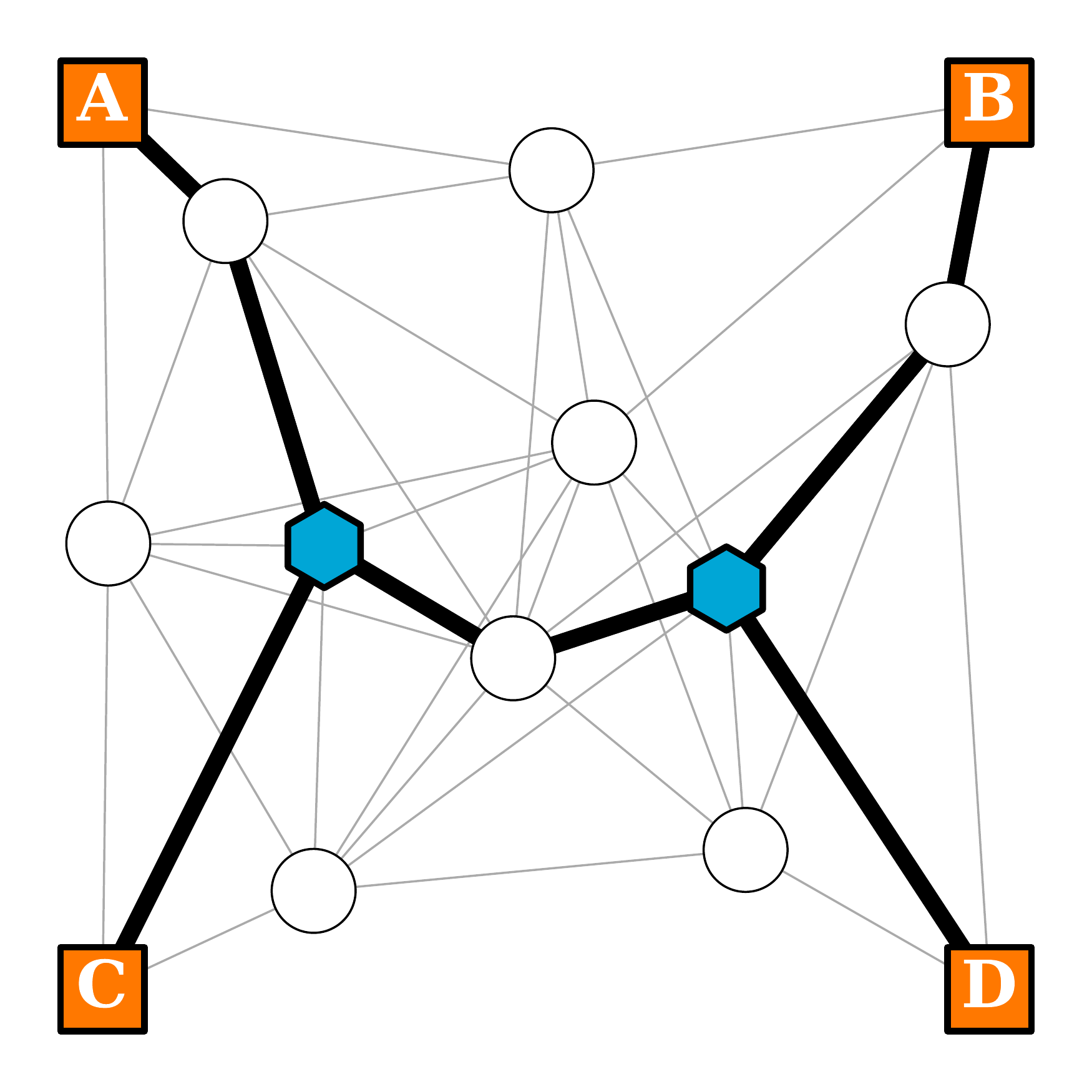}
		\caption{$L_\text{max} = 0.75$}
		\label{fig:square_Lmax075}
	\end{subfigure}
	~
	\begin{subfigure}[b]{0.3\textwidth}
		\includegraphics[width=\textwidth]{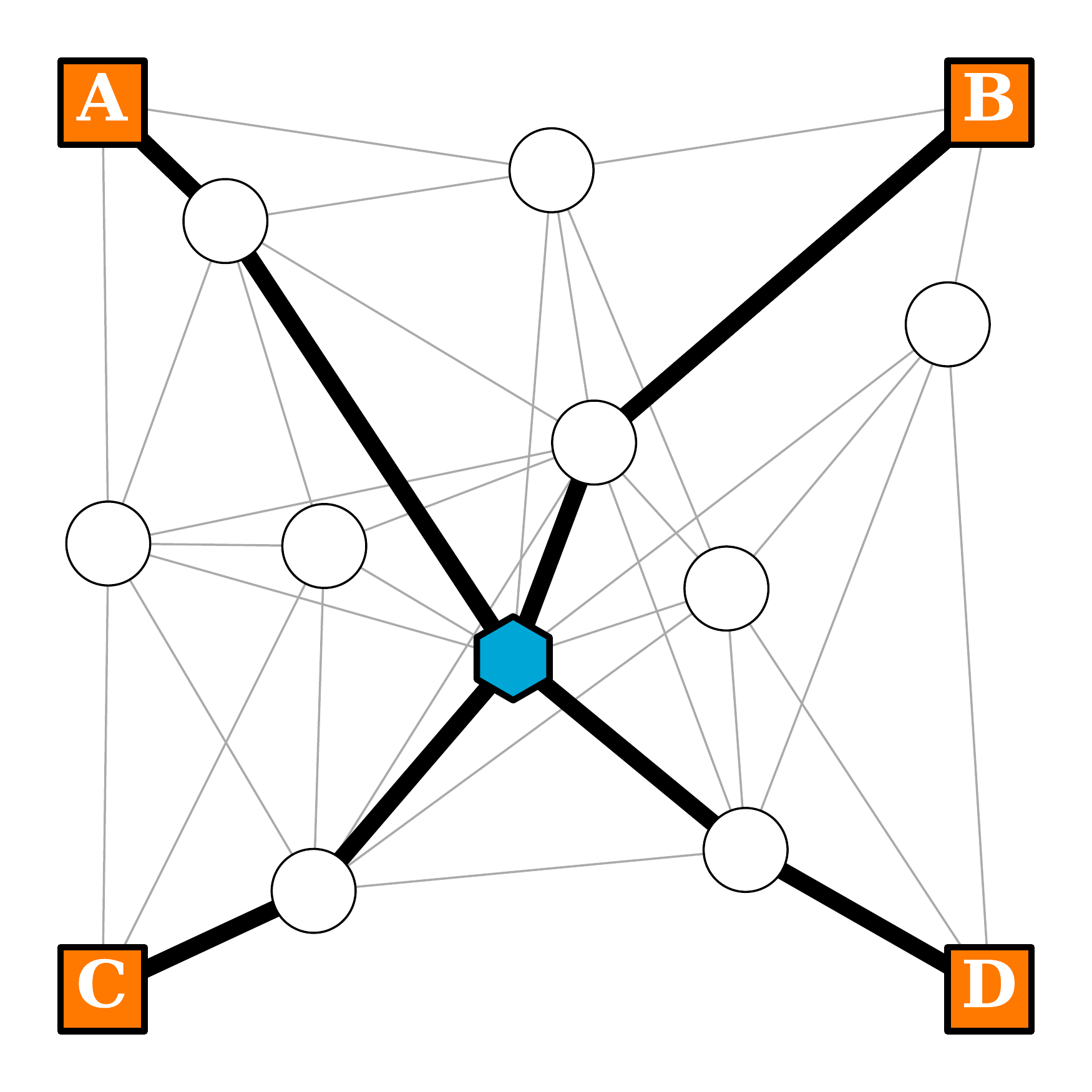}
		\caption{$L_\text{max} = 0.9$}
		\label{fig:square_Lmax09}
	\end{subfigure}
	
    \caption{
    Solutions obtained using our method for an example network graph using the network-requirement parameters
    $L_{\text{max}} = 0.9$, $N_{\text{max}} = 3$, $K = 1$ and $D = 6$, unless noted otherwise in the caption of a specific solution. 
    \textbf{(a)}-\textbf{(c)} Visualization of the effect of 
    $K$. 
    A higher robustness implies that we require more repeaters. 
    \textbf{(d)}-\textbf{(f)} 
    Visualization of the effect of $D$.
    As the
    capacity of quantum-repeater nodes increases,
    multiple paths can use the same repeater and hence the overall number of repeaters decreases.
    \textbf{(g)}-\textbf{(i)} Visualization of the effect of $L_\text{max}$.
    When longer elementary-link lengths are allowed, less quantum-repeater nodes are required to bridge the distance between end nodes.}
    \label{fig:unit_square_results}
\end{figure}

We here report how the number of placed repeaters and the (vertex) connectivity of quantum networks designed using our method vary as a function of the network-requirement parameters.
The number of placed repeaters is interesting to consider since the aim of our method is to minimize this.
On the other hand, the connectivity is interesting since it
lower bounds
the minimum number of quantum repeaters that need to break down before any pair of end nodes becomes disconnected, thereby giving an indication of how robust a quantum network is.
Note that connectivity is not the same as the robustness parameter $K$, which 
lower bounds
the 
minimum
number of quantum repeaters or elementary links that need to break down before end nodes can no longer distribute entanglement with a minimum rate and fidelity, while at the same time taking repeater capacity into account.
We have first generated 1000 random network graphs for which our method was able to find solutions for the parameter values $L_\text{max} = 0.9$, $N_\text{max}=6$, $K=6$ and $D=4$.
Then, while keeping all other parameters constant, we have varied each of the parameters $D$, $K$ and $L_\text{max}$.
This has been done in such a way that all considered values are less restrictive than the original values, such that we can be sure that a solution exists for each parameter value. 
Of each resulting quantum network, we determine the number of repeaters and the connectivity, and for each parameter value we determine the average number of repeaters and the average connectivity over all 1000 quantum networks.\\

\begin{figure}[b]
    \centering
    \captionsetup{justification=raggedright}
    \includegraphics[width=.85\textwidth]{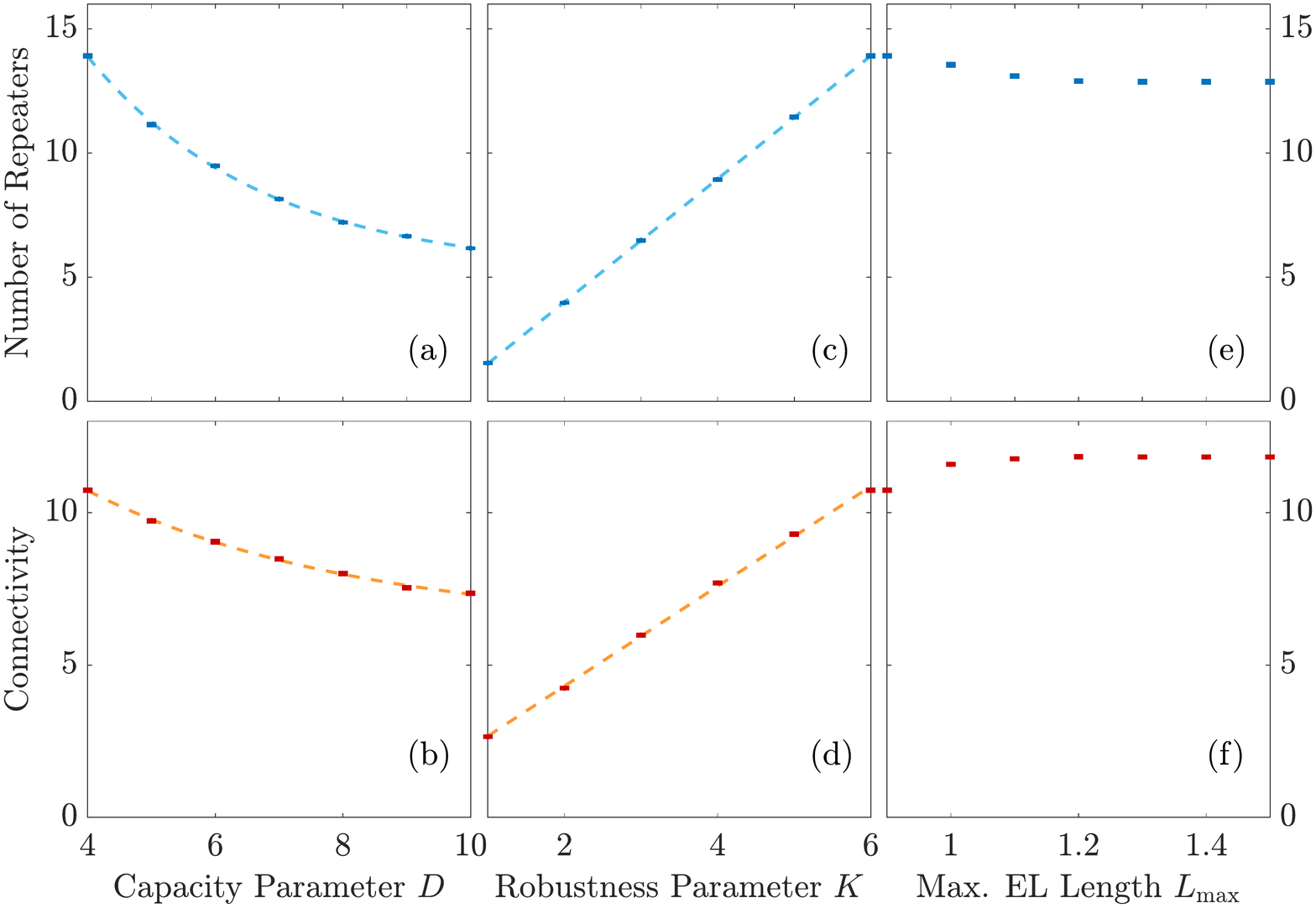}
    \caption{Simulation on 1000 random geometric graphs with a radius of $0.9$ and $n = 25$ nodes for varying network parameters. We use $L_{\text{max}} = 0.9$, $N_\text{max} = 6$, $K = 6$ and $D = 4$, and except for the varied parameter.
	In the plots, each of the points represents the average number of placed repeaters or average connectivity over all samples for each value of \textbf{(a)}-\textbf{(b)} the capacity parameter $D$, \textbf{(c)}-\textbf{(d)} the robustness parameter $K$ or \textbf{(e)}-\textbf{(f)} the maximum elementary-link length $L_{\text{max}}$. We either use a linear or an exponential function for the fits. The error bars represent one standard deviation of the mean. Solving an instance to optimality requires approximately $30$ seconds on average.}
	\label{fig:con_nrep_network_param}
	\addtocounter{figure}{-1}
	\begin{subfigure}{0\textwidth}
    \phantomsubcaption
    \label{fig:nrep_d}
    \end{subfigure}
    \begin{subfigure}{0\textwidth}
    \phantomsubcaption
    \label{fig:con_d}
    \end{subfigure}
    \begin{subfigure}{0\textwidth}
    \phantomsubcaption
    \label{fig:nrep_k}
    \end{subfigure}
    \begin{subfigure}{0\textwidth}
    \phantomsubcaption
    \label{fig:con_k}
    \end{subfigure}
    \begin{subfigure}{0\textwidth}
    \phantomsubcaption
    \label{fig:nrep_lmax}
    \end{subfigure}
    \begin{subfigure}{0\textwidth}
    \phantomsubcaption
    \label{fig:con_lmax}
    \end{subfigure}
\end{figure}

In \Cref{fig:nrep_d,fig:con_d}, we show the number of repeaters and the connectivity as a function of the repeater capacity $D$. We see that both the number of repeaters and the connectivity decrease as $D$ increases, and they both accurately follow an exponential fit in the domain under consideration.
In \Cref{fig:nrep_k,fig:con_k}, we show how the number of repeaters and connectivity vary as a function of the robustness parameter $K$. 
We see that both increase linearly in the domain under consideration. 
For $D$ ($K$) the number of repeaters decreases (increases) following the same line of reasoning as we mentioned above for the visual demonstration. Generally, we expect the connectivity to follow the change in the number of repeaters, because a network with less quantum repeaters is easier to disconnect.
Finally, in \Cref{fig:nrep_lmax,fig:con_lmax}, we investigate the effect of $L_\text{max}$ on the number of repeaters and connectivity.
While the number of repeaters decreases, the connectivity increases, although they both flatten from $L_{\text{max}} = 1.2$. The number of repeaters does not decrease to zero because $K = 6$.
Therefore, even if $L_\text{max}$ is large enough to allow for paths between end nodes with zero quantum-repeater nodes, there are still at least five quantum-repeater nodes required to make the network robust against the breakdown of direct elementary links between end nodes.
On the other hand, the connectivity increases since it also takes paths through other end nodes into account in its computation, and with an increasing value of $L_\text{max}$, we expect more direct elementary links to appear. 

\subsection{Computation Times}
\label{subsec:discussion:comp_times}

Even though the link-based formulation has a scaling of $O(|\mc N|^2)$ in terms of the number of variables and constraints, it remains an ILP. In general, ILP's are NP-hard and thus generally require an exponential amount of time to solve. 
In order to investigate the performance of our method for varying network sizes, we determined the computation time for finding an optimal solution as a function of the number of nodes. The result is shown in \Cref{fig:comp_times}, in which we see that the computation time indeed increases exponentially. Nonetheless, instances on random geometric graphs with 100 nodes can be solved to optimality in about one minute on our setup (see \Cref{sec:results}). \\

The computation time can be strongly affected by the network topology and the chosen parameter values, since these can alter the difficulty of finding an optimal solution as well as the number of variables and constraints (see \Cref{subsec:methods:scaling}). However, the parameters that we use for \Cref{fig:comp_times} are neither very strict nor loose and provide us with insight into the approximate scaling of the computation time, rather than the worst-case behavior.
Note that we expect that, in practical use cases, the topology and the parameter values will be determined once and remain more or less fixed, which implies that the repeater-allocation problem will not need to be solved repeatedly. This makes the increasingly large computation time for sizable graphs or stringent parameters less problematic.

\begin{figure}[b]
    \centering
    \captionsetup{justification=raggedright}
    \vspace*{-50pt}
    \includegraphics[width=0.75\textwidth]{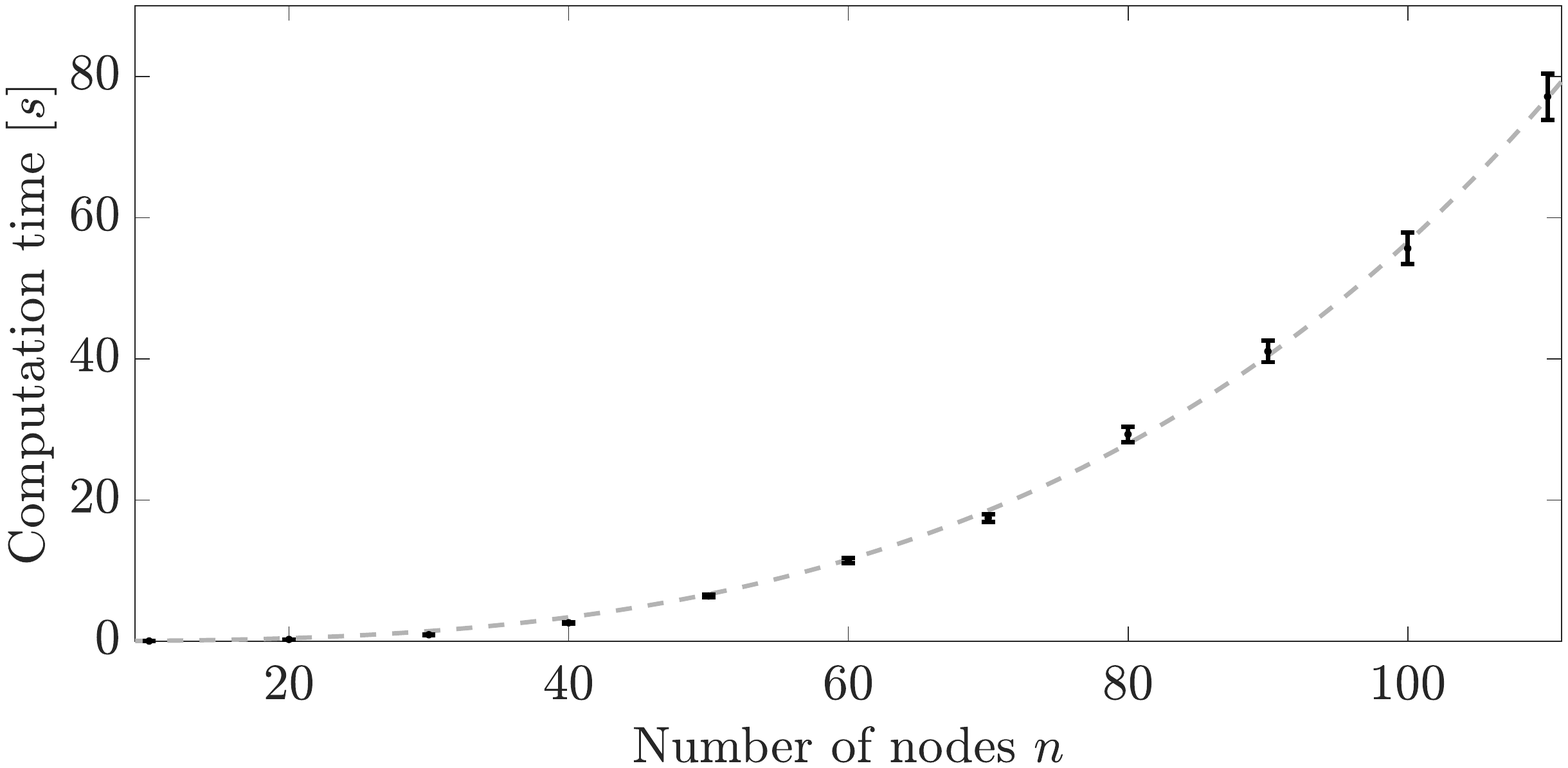}
    \vspace*{-50pt}
    \caption{Computation time in seconds for 100 random geometric graphs with $L_{\text{max}} = 1$, $N_{\text{max}} = 6$, $K = 2$ and $D = 8$ for varying number of nodes $n$. The error bars represent one standard deviation of the mean. For the fit we have used an exponential function of the form $a(e^{bn^3} - 1)$, where $a$ and $b$ are free parameters.}
    \label{fig:comp_times}
\end{figure}

\subsection{Extensions}
\label{subsec:discussion:extensions}

There are various ways in which our method can be extended.
Here, we present two possible extensions.
Such extensions change the ILP formulation in \Cref{lbf}.
The result of these is the generalized link-based formulation, which is presented in \Cref{generalized_lbf}. 
To incorporate the extensions into the method in \Cref{method}, the generalized link-based formulation must be used where otherwise the link-based formulation would be used. \\

\begin{table}[ht]
\boxalign{
	\begin{alignat}{2}
	\hspace*{-100pt}
	\min \quad \sum_{u \in \mathcal{R}} y_u & + \phantom{\leq}  \alpha \sum_{q\in\mathcal{Q}} \sum_{(u,v) \in \mathcal E_q} \sum_{k=1}^{K^q} L\Big( (u, v) \Big)x^{q,k}_{uv} \label{eq:ext_lbf_objfun} \\
	\mathrm{s.t.} \quad 
	\sum_{\substack{v \\ (u, v) \in \mathcal E _q}} x^{q,k}_{uv} - \sum_{\substack{v \\ (v, u) \in \mathcal E _q}}x^{q,k}_{vu} & =
	\begin{cases}
	1, \quad & \text{if } u = s \\
	-1, \quad & \text{if } u = t \\
	0, \quad & \text{if } u \in \mathcal{R}
	\end{cases}
	\qquad
	&&\forall u \in \mathcal R \cup \{s, t\}, q = (s, t) \in \mathcal{Q}, k = 1, 2, \ldots, K^q \label{eq:ext_lbf_flowcon} \\
	L\Big( (u, v) \Big) x^{q,k}_{uv} & \leq L^q_{\mathrm{max}} &&  \forall (u,v) \in \mathcal E_q, q \in \mathcal{Q}, k = 1, 2, \ldots, K^q \label{eq:ext_lbf_max_length} \\
	\sum_{(u,v) \in \mathcal E _q}  x^{q,k}_{uv} & \leq N^q_{\mathrm{max}} +1 && \forall q \in \mathcal{Q}, k = 1, 2, \ldots, K^q \label{eq:ext_lbf_max_repeaters} \\
	\sum_{\substack{v \\ (u, v) \in \mathcal E_q}}\sum_{k=1}^{K^q} x_{uv}^{q,k} & \leq 1 && \forall u \in \mathcal{R}, q \in \mathcal{Q} \label{eq:ext_lbf_disjoint_repeaters} \\
	\sum_{k=1}^{K^q} x_{st}^{q,k} & \leq 1 && \forall q \in \mathcal{Q} \label{eq:ext_lbf_max_one_direct_path} \\
	\sum_{q \in\mathcal{Q}} \sum_{\substack{v \\ (u, v) \in \mathcal E _q}} \sum_{k = 1}^{K^q} x^{q,k}_{uv} & \leq D_u y_u && \forall u \in \mathcal{R} \label{eq:ext_lbf_capacity} \\
	x^{q,k}_{uv} & \in \{0, 1\} && \forall (u,v) \in \mathcal{E}_q, q \in \mathcal{Q}, k = 1, 2, \ldots, K^q \label{eq:ext_xdmon_lbf_k} \\
	y_u & \in \{0, 1\} && \forall u \in \mathcal{R} \label{eq:ext_ydom_lbf}
	\end{alignat}
}
\caption{Generalized link-based formulation.}
\label{generalized_lbf}
\end{table}

The first extension we can make is solving the repeater-allocation problem in case of heterogeneous network requirements.
So far, we have considered the network requirements to be homogeneous, i.e. the same throughout the network.
However, it can be the case that some end nodes require a higher rate and fidelity, that some end nodes need access to more robust quantum communication, or that quantum repeaters with a larger capacity can be placed at some potential repeater locations than at other.
Then, we can define the network-requirement parameters on a per-end-node-pair or per-node basis.
Specifically, for every pair of end nodes $q \in \mathcal Q$, 
we define the minimum rate $R_\text{min}^q$ and fidelity $F_\text{min}^q$ of entanglement generation,
and the required robustness parameter $K^q$ (in order to break communication between the end nodes $q$, at least $K^q$ quantum repeaters or elementary links must be incapacitated).
Furthermore, for every potential repeater location $u \in \mathcal R$, we define the quantum-repeater capacity $D_u$.
To incorporate this into the method, the input parameters must be adapted accordingly, and the maximum number of repeaters and maximum elementary-link length must be calculated for every pair of end nodes separately (i.e. $L_\text{max}^q$ and $N_\text{max}^q$ must be determined from $R_\text{min}^q$ and $F_\text{min}^q$ for each $q \in \mathcal Q$).\\

A second extension has to do with the fact that the link-based formulation in \Cref{lbf} typically has a highly-degenerate optimal solution.
That is, often there are multiple possible quantum-repeater placements for which all constraints are satisfied and the total number of quantum-repeater nodes is minimal.
However, it might be the case that some solutions are more desirable than others.
To pick out these solutions, one can define a secondary objective.
This secondary objective can then be taken into account by defining a corresponding objective function, and adding it to the existing objective function, while scaling it such that it does not influence the optimal number of repeaters. In particular, the scale factor $\alpha$ should be chosen such that the secondary objective value does not exceed $1$.
This can be seen as a form of weighted goal programming \cite{jones2010practical}.
As an example, in \Cref{generalized_lbf}, we use as secondary objective to minimize the total length of all used elementary links. 
Other secondary objectives, such as minimizing the largest elementary-link length, could be implemented in a similar fashion. \\

\subsection{Limitations}
\label{subsec:discussion:limitations}

In this section we discuss some of the limitations of the method we present in this work.
Each limitation represents a way that our method could be further extended, but is beyond the scope of this paper.\\
                                              
A first major limitation is the complexity of ILP's.
While we provide an efficient ILP formulation, 
in which the number of variables and constraints scales polynomially with the network size,
it remains an ILP.
This cannot be helped, as choosing whether a repeater should be placed at a certain potential repeater location is inherently binary.
In general, it is NP-hard to solve an ILP.
While we indeed observe exponential scaling of the computation time in \Cref{subsec:discussion:comp_times}, we are able to find optimal solutions of realistically-sized networks within tractable time using CPLEX, which is also demonstrated using a real network in \Cref{subsec:discussion:real_network}.
Conceivably, one can use heuristics or approximation algorithms to obtain solutions faster, although the solutions then may no longer be optimal. \\

Another limitation that we consider here is the fact that our method is agnostic about how elementary links are constructed.
We assume that any number of fibers can be combined to form an elementary link.
However, quantum-repeater protocols relying on heralded entanglement generation typically require the presence of a midpoint station with the capability to perform Bell-state measurements \cite{inside_quantum_repeaters}.
If there are constraints on the placement of such stations, our method is insufficient.
Conceivably, if such stations can only be placed at potential repeater locations, a modified version of our method could be used.
Furthermore, we assume that an elementary link between two nodes is always constructed from the fibers which minimize the elementary-link length such that rate and fidelity are maximized.
However, if one would like to incorporate the number of fibers (rather than elementary links) that need to be disabled before the quantum network is incapacitated as an additional network requirement (thereby guaranteeing more robustness),
this may no longer be a useful assumption.
It may then be better to try to construct different elementary links from different fibers as much as possible, such that individual fibers do not become too critical. \\

%% file: methods.tex
\section{Methods}
\label{sec:methods}

\subsection{Explanation of the Path-Based Formulation}
\label{subsec:methods:validity_of_pbf}

In \Cref{subsec:results:pbf}, we introduced the path-based formulation.
This ILP formulation can be found in \Cref{pbf}, and we claim that solutions to the path-based formulation can be used to construct solutions to the repeater-allocation problem.
Here, we show how and why this can be done. \\

The idea behind the path-based formulation is to choose a combination of feasible paths that minimize the overall number of utilized repeaters. If a path is chosen that uses potential repeater location $u \in \mathcal R$ as a quantum-repeater node,
a repeater should be placed at $u$.
The binary variables $x_p$ are used to parameterize the chosen paths, while the binary variables $y_u$ are used to parameterize where quantum repeaters should be placed.
A coupling between these variables is realized by Constraints \eqref{eq:pbf_capacity}: if a path $p \in \mc P$ is chosen in which a node $u \in \mc R$ is used as quantum-repeater node, the corresponding $y_u$ variables must have value $1$. Conversely, when $y_u = 1$ for a given repeater node $u \in \mc R$, up to $D$ paths can use this repeater node in order for the corresponding constraint to hold, thereby also imposing a limit on the repeater capacity. After all, if $\sum_{p \in \mathcal P}r_{up}x_p > D$ then more than $D$ paths are chosen in which node $u \in \mc R$ is used as a repeater, which renders the solution infeasible. \\

Paths are moreover only considered useful if they can be used to deliver entanglement between end nodes with the minimum required rate $R_\text{min}$ and fidelity $F_\text{min}$.
In the path-based formulation, this is implemented by requiring chosen paths to contain at most $N_\text{max} + 1$ elementary links, each with a length of at most $L_\text{max}$. 
The values of $N_\text{max}$ and $L_\text{max}$ can be determined from $R_\text{min}$ and $F_\text{min}$ as detailed in \Cref{method}.
These requirements are straightforwardly enforced by Constraints \eqref{eq:pbf_max_length} and \eqref{eq:pbf_max_repeaters}.
Constraints \eqref{eq:pbf_max_length} can only hold when $x_p = 0$ for all paths that contain an elementary link ($(u, v) \in p$) which is too long ($L((u, v)) > L_\text{max}$).
Similarly, Constraints \eqref{eq:pbf_max_repeaters} can only hold when $x_p = 0$ for all paths for which the number of elementary links ($|p|$) exceeds the maximum ($N_\text{max} + 1$).\\

Furthermore, the choice of paths must be such that it is guaranteed that up to $K$ potential repeater nodes or elementary links can break down before there is no path available between any pair of end nodes that can deliver entanglement at the required rate and fidelity.
This is implemented by choosing, per pair of end nodes, $K$ different paths.
All of these paths are chosen such that none of them share a quantum-repeater node.
Since elementary links connect quantum-repeater nodes, this automatically also means that none of the paths share an elementary link.
Therefore, when a quantum-repeater node or elementary link becomes incapacitated, this can disrupt at most one path between a pair of end nodes.
When there are $K$ break downs, in the worst case, this can disrupt all paths between a pair of end nodes.
But as long as there are fewer break downs, there will be at least one path available.\\

Since every chosen path can deliver entanglement at the required rate and fidelity, this guarantees robustness of the quantum network against up to $K$ break downs.
It is enforced by Constraints \eqref{eq:pbf_K} that there are exactly $K$ paths chosen between every pair of end nodes.
Furthermore, Constraints \eqref{eq:pbf_disjoint} make sure that the number of chosen paths connecting a pair of end nodes using $u$ as a quantum-repeater node ($\sum_{p \in \mathcal P_q} r_{up}x_p$) is at most one, thereby guaranteeing that all $K$ paths are disjoint. 
Note that, when considering the quantum-repeater capacity, all chosen paths are taken into account.
In other words, Constraints \eqref{eq:pbf_capacity} guarantee that the repeater capacity is not exceeded when all paths are used simultaneously.
Therefore, if one path between a pair of end nodes is disrupted and they are forced to switch to another path, it is guaranteed that none of the quantum repeaters along that path are overloaded.
\\

It is now easy to obtain a solution to the repeater-allocation problem from the solution to the path-based formulation.
Every potential repeater location $u \in \mathcal R$ for which $y_u=1$ in the solution to the path-based formulation should be used as a quantum-repeater node.
Furthermore, each elementary link which is part of a chosen path ($(u, v) \in p$ such that $x_p=1$)  should be constructed.
This is done 
using the fibers making it up ($F\big((u,v)\big)$).
Then, the resulting quantum network will be such that all network requirements are satisfied.
Furthermore, the number of quantum-repeater nodes will be minimal.
This is because this number, which is exactly $\sum_{u \in \mathcal R}y_u$, is minimized by the objective function \eqref{eq:pbf_objfun} of the path-based formulation.
Therefore, the path-based formulation can indeed be used to solve the repeater-allocation problem.

\subsection{Toy-Model Calculation of \texorpdfstring{$N_\text{max}$ and $L_\text{max}$ from $R_\text{min}$ and $F_\text{min}$}{N\_max and L\_max from R\_min and F\_min}}
\label{subsec:methods:toy-model}

In this section we calculate the maximum number of repeaters and maximum elementary-link length from the minimum required rate $R_\text{min}$ and fidelity $F_\text{min}$ using a toy model of a quantum-repeater chain.
The quantum-repeater architecture that we consider is of the massively-multiplexed type as described in e.g. \cite{sinclairSpectralMultiplexingScalable2014}.
The toy model that we consider here makes the following simplifying assumptions:
\begin{itemize}
    \item the states distributed over elementary links are Werner states,
    \item the noise in the states distributed over elementary links is the only noise,
    \item the only sources of photon loss are fiber attenuation and non-deterministic Bell-state measurements,
    \item all processes except light traveling through fiber are instantaneous.
\end{itemize}

It is shown in \Cref{app:toy_model} that in this model, a repeater chain with $N$ quantum repeaters, $M$ entanglement-distribution attempts per round per elementary link, elementary-link length $L$, elementary-link fidelity $F_\text{link}$, speed of light in fiber $c_\text{fiber}$ and a 50\% Bell-state measurement success probability has the following end-to-end rate $R$ and fidelity $F$:
\begin{align}
        R &= \frac {c_\text{fiber}} L \Big( \frac 1 2 \Big)^N \Big[ 1  - \Big( 1 - \frac 1 2 e^{-L/L_\text{att}}\Big)^M \Big]^{N+1}, \\
        F &= \frac 1 4 \Big[ 1 + 3 \Big ( \frac{4F_\text{link} - 1} 3 \Big )^{N+1} \Big].
\end{align}
$N_\text{max}$ can now be obtained from the fidelity.
Specifically, it is the lowest-integer solution to the equation
\begin{equation}
    F > F_\text{min}.
\end{equation}
To find $L_\text{max}$, we can put the resulting value of $N_\text{max}$ into the equation
\begin{equation}
    R > R_\text{min}.
    \label{eq:minimum_rate_eq}
\end{equation}
The smallest value for $L$ that solves Equation \eqref{eq:minimum_rate_eq} is then $L_\text{max}$.
Note that the calculation here is somewhat simplified because the fidelity is not a function of $L_\text{max}$.
If both fidelity and rate would be functions of $N$ and $L$, there would not exist a unique solution.
In that case, there is some freedom in choosing $N_\text{max}$ and $L_\text{max}$.\\

The calculation of $N_\text{max}$ and $L_\text{max}$ for the example parameters $F_\text{min} = 0.93$, $R_\text{min} = 1$ Hz, $F_\text{link} = 0.99$, $c_\text{fiber} = 200,000$ km/s, $M=1000$ and $L_\text{att} = 22$ km results in $N_\text{max} = 6$ and $L_\text{max} = 136$ km (rounded down).

\subsection{Proof of Equivalence}
\label{subsec:methods:proof_of_equivalence}

In this section we briefly outline the proof of why the path-based formulation and the link-based formulation are equivalent. 
The main idea is to use an optimal solution to the path-based formulation to construct a feasible solution to the link-based formulation and vice versa. 
We prove that this is always possible in such a way that the value of the objective function of the constructed feasible solution is the same as that of the original optimal solution.
This can be used to show that the optimal objective values of both formulations are always the same. 
Therefore, the feasible solution to one formulation constructed from an optimal solution to another formulation is itself an optimal solution.
We say that two ILP formulations are equivalent if optimal solutions to one can be obtained from the other and vice versa, and therefore
we conclude that the path-based formulation and the link-based formulation are equivalent. \\

To construct a solution to the link-based formulation using a solution to the path-based formulation, we use the elementary links that appear in chosen paths. More specifically, for each $q = (s, t) \in \mc Q$ and $k = 1, 2, \dots, K$, we set $x_{uv}^{q,k} = 1$ if elementary link $(u, v) \in \mc E_q$ is in the $k^\text{th}$ chosen path connecting $s$ and $t$. Conversely, Constraints \eqref{eq:lbf_flowcon} guarantee that, for every $q = (s, t) \in \mc Q$ and $k = 1, 2, \dots, K$, the elementary links $(u, v) \in \mc E_q$ for which $x_{uv}^{q,k} = 1$ can be used to form exactly one path between $s$ and $t$. These paths can be obtained by using \Cref{alg:path_extraction}, which outputs the set $\mathcal P^*$ that contains the extracted paths over all $q \in \mathcal Q$ and $k = 1, 2, \dots, K$. Thus, we can construct a solution to the path-based formulation from a solution to the link-based formulation by setting $x_p = 1$ for all $p \in \mc P^*$. Furthermore, the repeater-placement variables $y_u$ are kept the same when translating between formulations. \\

By comparing the different constraints, it can be understood that if a solution to one formulation is feasible, the solution to the other formulation that can be obtained from it is also feasible.
Constraints \eqref{eq:pbf_max_length} and \eqref{eq:lbf_max_length} both guarantee that elementary-link lengths do not exceed $L_\text{max}$, while
Constraints \eqref{eq:pbf_max_repeaters} and \eqref{eq:lbf_max_repeaters} both guarantee that each path includes $N_\text{max}$ quantum-repeater nodes at maximum.
Constraints \eqref{eq:pbf_K} and \eqref{eq:lbf_flowcon} make sure there are $K$ paths between each pair of end nodes. These paths are guaranteed to be disjoint for the path-based formulation by Constraints \eqref{eq:pbf_disjoint} and for the link-based formulation by Constraints \eqref{eq:lbf_disjoint_repeaters} and \eqref{eq:lbf_max_one_direct_path}.
Lastly, Constraints \eqref{eq:pbf_capacity} and \eqref{eq:lbf_capacity} couple the $x$ variables to the $y$ variables and make sure the quantum-repeater capacity is taken into account. \\

In step \ref{step:remove_loops} of \Cref{method}, we manually set $x_{uv}^{q,k}=0$ for all elementary links $(u,v) \in \mc E$ which are not in one of the paths $p \in \mathcal P^*$.
We do this because, on some occasions, the variables $x^{q,k}_{uv}$ are allowed to have value $1$ in such a way that they form loops (which are disjoint from the path between $s$ and $t$).
For example, it could be the case that for some $q \in \mathcal Q$ and $k = 1,2, \dots, K$, it holds that $x^{q,k}_{u_1 u_2} = x^{q,k}_{u_2 u_1} =1$, which does not violate any of the constraints in \Cref{lbf}, and also does not influence the objective function \eqref{eq:lbf_objfun}.
Since these loops do not connect end nodes, they do not contribute to realizing any of the network requirements.
Therefore, any variable $x^{q,k}_{uv}$ with value $1$ such that it is part of a loop can safely be set to $0$ without violating any constraint.
This is shown rigorously in \Cref{app:proof_of_equivalence}.
Only allowing for elementary links which are part of paths between end nodes realizes the removal of such loops.
Since the method in \Cref{method} recommends the construction of elementary link $(u,v) \in \mc E$ if $x_{uv}^{q,k} = 1$, setting them to $0$ whenever this is possible helps to prevent the construction of unnecessary elementary links.
One way in which the appearance of loops in optimal solutions can be prevented in the first place by is to use the generalized link-based formulation in \Cref{generalized_lbf}.
In this formulation, the minimization of the total elementary-link length is used as secondary objective. \\

\smallskip
\begin{algorithm}[b]
	\SetAlgoLined
	$\mc P^* = \emptyset$\;
	\For{$q = (s, t) \in \mc Q$}{
	    \For{$k = 1, 2, \ldots, K$}{
	        $u_0 = s$\;
	        $n = 0$\;
	        \While{$u_n \neq t$}
	        {
	            Find the unique node $v \in \mc R \cup \{t\}$ for which $x_{u_nv}^{q,k} = 1$\;
                $n = n + 1$\;
                $u_n = v$\;
	        }
	        $p = \Big( (s, u_1), (u_1, u_2), \dots , (u_{n-1}, t) \Big)$ \;
            $\mc P^* = \mc P^* \cup p$\;
	    }
	}
	\caption{Path extraction algorithm.}
	\label{alg:path_extraction}
\end{algorithm}
\smallskip

\subsection{Generating Random Networks}
\label{subsec:methods:generating_random_networks}

Here, we describe how we generate random network graphs based on random geometric graphs.
These networks are used to demonstrate our method and study the effect of different network-requirement parameters in \Cref{subsec:discussion:evaluation_on_random_networks}. \\

The recipe for generating a random geometric graph on a two-dimensional Euclidean space with $n$ nodes and radius $d$
is as follows \cite{Pen03}.
First, $n$ points are distributed uniformly at random on a unit square, by sampling both their horizontal and vertical coordinates uniformly at random.
To every two points $p_1$, $p_2$ we associate $r(p_1, p_2)$, which is the Euclidean distance between the two points.
From this, an undirected weighted graph is constructed in which every node corresponds to one of the points, 
and edges between nodes corresponding to points $p_1$, $p_2$ are added if $r(p_1, p_2) \leq d$.
The weight that is given to the edge is $r(p_1, p_2)$. \\

To turn a random geometric graph into a suitable network graph, 
it must be decided which of the nodes are end nodes, and which are potential repeater locations.
To this end, we determine the convex hull of the graph.
We choose to use nodes corresponding to vertices of the convex hull of the graph as end nodes, i.e. they make up the set $\mathcal C$.
All other nodes are thus considered potential repeater locations, i.e. they make up the set $\mathcal R$.
This method is used because it is expected that potential repeater locations lying outside of the area spanned by the end nodes will only rarely be chosen as quantum-repeater nodes.
When the end nodes form the convex hull, there are no such potential repeater locations, and the number of nodes that are not of relevance to the repeater-allocation problem is minimized.
We generate the random geometric graphs using NetworkX \cite{networkx} and determine the convex hull using an algorithm \cite{qhull} which is included in SciPy \cite{scipy}.\\

The random network graph used in \Cref{fig:unit_square_results} has been based on a random geometric graph with $n=10$ and $d=0.6$, but has been further edited to be made suitable for demonstration purposes.
Some nodes were displaced manually.
Additionally, end nodes have been added at the corners of the unit square and connected to the three closest potential repeater locations.

\subsection{Scaling of the Formulations}
\label{subsec:methods:scaling}

The path-based formulation relies on the enumeration of all the paths between two end nodes. For every pair $(s,t) \in \mc Q$ we must consider all possible permutations of intermediate nodes in which $r$ repeaters are placed on a path. For $r = 0$, we get a single path directly from $s$ to $t$ and for $r = 1$ we should consider all possible paths that utilize one repeater, which are $|\mc R|$ in total. Next, when $r = 2$ we must consider all paths that contain exactly two repeaters and additionally all permutations of the repeater placements in these paths, which gives $|\mc R|(|\mc R| - 1)$ paths in total, et cetera. The number of $y_u$ variables is $|\mc R|$, so that the number of variables $n_\text{var}^\text{pbf}$ of the path-based formulation is given by
\begin{align}
    n_\text{var}^\text{pbf} & = |\mc R| + |\mc Q||\mc P_q| \\
    & = |\mathcal{R}| + |\mathcal{Q}|\sum_{r=0}^{|\mc R|}\frac{|\mathcal{R}|!}{(|\mathcal{R}| - r)!}. \label{eq:scaling_pbf}
\end{align}
If $|\mc R| > 1$, this simplifies to \cite{wagon2016round}
\begin{align}
    n_\text{var}^\text{pbf} & = |\mathcal{R}| + |\mathcal{Q}|\left[e|\mc R|!\right], 
\end{align}
where $e$ denotes Euler's number and $[\cdot]$ represents the rounding operator. We assume that the number of end nodes $|\mc C|$, and therefore the number of end-node pairs $|\mc Q| = |\mathcal{C}|(|\mathcal{C}| - 1) / 2$, is constant so that this does not scale with the total number of nodes $|\mc N|$ in our graph. This implies that the number of possible repeater locations $\mc R = \mc N \setminus \mc C$ scales linearly with the number of nodes. The number of variables, as well as the number of constraints, is thus $O\left(|\mc N|!\right)$. \\

One important detail of our implementation of the path-based formulation is that we take Constraints \eqref{eq:pbf_max_length} and \eqref{eq:pbf_max_repeaters} into account while enumerating all the paths. If we encounter a path which contains an elementary link with a length that exceeds $L_\text{max}$ or which uses more than $N_\text{max}$ repeaters, we simply exclude it from the set $\mc P$. This can greatly reduce the total number of variables, although it will remain to scale exponentially with $|\mc N|$. \\

In the link-based formulation, we need to enumerate all the elementary links in the network. To this end, we need to count every elementary link from $s$ to every node $v \in \mathcal{R} \cup \{t\}$, and from $u \in \mc R$ to $t$ which results in $2|\mc R| + 1$ elementary links. Next, we also need to consider the elementary link from every node $u \in \mc R$ to $v \in \mc R$ and back, in order to allow for directional paths from $s$ to $t$, which are $|\mc R|(|\mc R| - 1)$ in total. Additionally, since we use the index $k$ for our $x_{uv}^{q,k}$ variables in order to keep track of the redundant paths that are required for the given level of robustness, we need to make a copy of these variables for every value of $k = 1, 2, \ldots, K$. When we combine this with the $|\mc R|$ $y_u$ variables, we get that the total number of variables of the link-based formulation is given by

\begin{align}
    n_\text{var}^\text{lbf} & = |\mc R| + K|\mc Q||\mc E_q| \\
    & = |\mc R| + K|\mc Q|(|\mc R|^2 + |\mc R| + 1), \label{eq:scaling_lbf}
\end{align}

\noindent which is $O\left(|\mathcal{N}|^2\right)$, if we assume that $K$ is a fixed constant.
Note that the link-based formulation therefore also has $O\left(|\mc N|^2\right)$ constraints.

%% file: acknowledgements.tex
\section{Acknowledgements}

We would like to thank SURFnet for sharing their data regarding the network topology. This publication was supported by the QIA-project that has received funding from the European Union's Horizon 2020 research and innovation program under grant Agreement No. 820445. This work was supported by an ERC Starting grant and NWO Zwaartekracht QSC.

%% file: contributions.tex
\section*{Data Availability}
All the data and code we used for generating the results can be found in the Github repository \cite{githubcode}.

\section*{Author Contributions}
This work is based on the master thesis of J.R. which was devised by S.W. and supervised by S.W. and G.A.; K.C. introduced some key conceptual and proof ideas. J.R. proposed the use of linear programming and implemented both formulations with CPLEX and G.A. wrote most of the code for the simulations. All authors contributed to the manuscript.

\section*{Competing Interests}
The authors declare no competing interests.

%% file: appendix.tex
\section{Toy-Model Calculation of Rate and Fidelity}
\label{app:toy_model}

In this appendix, we calculate the rate and fidelity of a quantum-repeater chain using a toy model described in \Cref{subsec:methods:toy-model}.
The quantum-repeater architecture under consideration is of the massively-multiplexed type as described in e.g. \cite{sinclairSpectralMultiplexingScalable2014}.
In such a repeater chain, during every round of time, entanglement distribution is attempted a large number of times on each elementary link (using e.g. spectral multiplexing).
If at the end of the round a quantum repeater has at least succeeded once at entanglement generation with each neighbour, a successfully-entangled state is selected from each side and entanglement swapping is performed between the two (through a Bell-state measurement).
Otherwise, all entanglement is discarded and a new attempt is made during the next round.
Our toy model of such a quantum-repeater chain is based on the following simplifying assumptions:
\begin{itemize}
    \item the states distributed over elementary links are Werner states,
    \item the noise in the states distributed over elementary links is the only noise,
    \item the only sources of photon loss are fiber attenuation and non-deterministic Bell-state measurements,
    \item all processes except light traveling through fiber are instantaneous.
\end{itemize}

First, we investigate the final end-to-end fidelity of entangled quantum states created by a repeater chain.
Let us consider a repeater chain with $N$ quantum repeaters, $M$ entanglement-distribution attempts per round per elementary link, elementary-link length $L$ and elementary-link fidelity $F_\text{link}$.
In the toy model, entangled states shared over elementary links are Werner states, which can be parametrized as
\begin{equation}
    \rho_{p_\text{link}} = p_\text{link} \ketbra {\Phi^+} + \frac{1-p_\text{link}}{4} \mathds 1.
\end{equation}
This state has fidelity to the maximally-entangled Bell state $\ket {\Phi ^+} = \tfrac 1 {\sqrt 2} (\ket {00} + \ket {11})$ of $F_\text{link} = \tfrac 1 4 (1 + 3 p_\text{link})$, and therefore $p_\text{link} = \tfrac 1 3 (4F_\text{link}-1)$.\\

Entanglement swapping between Werner states $\rho_{p_1}$ and $\rho_{p_2}$ is performed through a Bell-state measurement on one qubit from the first state and one qubit from the second state. 
The quantum state after this operation (after tracing out the measured qubits and ignoring possible Pauli corrections) is a new Werner state, $\rho_{p_{1,2}}$, with $p_{1,2}$ = $p_1 p_2$.
Repeated use of this equation reveals that, if entanglement distribution is successful at least once in each of the $N+1$ elementary links, and if all entanglement swaps are successful, the Werner state $\rho_{p_f}$ is obtained with $p_f = p_\text{link}^{N+1}$.
Thus, the final fidelity is
\begin{equation}
    F = \frac {1 + 3 p_\text{link}^{N+1} } 4 = \frac 1 4 \Big[ 1 + 3 \Big ( \frac{4F_\text{link} - 1} 3 \Big )^{N+1} \Big].
\end{equation}

Now we consider the rate at which end-to-end entanglement can be established.
First, we calculate the probability that a single attempt at entanglement distribution in a single elementary link is successful.
Entanglement is generated by sending entangled photons from both repeaters to a station in the center of the elementary link, where a probabilistic Bell-state measurement is performed with 50\% success probability.
In the toy model, the only other source of photon loss in the elementary link is attenuation in optical fiber, which we assume to be characterized by the attenuation length $L_\text{att}$.
In that case, the probability that both photons reach the midpoint and the Bell-state measurement in successful is

\begin{equation}
    \text{Pr(one attempt)} = \Big( e^{-L/(2L_\text{att})}\Big)^2 \times \frac 1 2 = \frac 1 2 e^{-L/L_\text{att}}.
\end{equation}
Then, the probability that at least one of the $M$ attempts in an elementary link in a single round is successful is
\begin{equation}
    \text{Pr(elementary link)}  = 1 - \left(1 - \text{Pr(one attempt)}\right)^M.
\end{equation}
Finally, the end-to-end success probability is given by the probability that each link is successful, and that each entanglement swap in the repeaters is successful.
Since entanglement swapping in repeaters has a 50\% success probability, this gives
\begin{equation}
    \text{Pr(repeater chain)}  = \text{Pr(elementary link)} ^{N+1} \Big(\frac 1 2 \Big)^N =  \Big( \frac 1 2 \Big)^N \Big[ 1  - \Big( 1 - \frac 1 2 e^{-L/L_\text{att}}\Big)^M \Big]^{N+1}.
\end{equation}

To determine the rate, we now need to know how long every round takes.
In the toy model, the only aspect of entanglement generation that takes any time is the photons traveling to the midpoints stations, and the messages heralding success or failure of the entanglement attempts traveling from the midpoint stations back to the repeaters.
We assume both are light traveling through fiber.
Thus, every round takes as long as it takes for light to travel the distance $L$ through fiber.
Denoting the speed of light in fiber $c_\text{fiber}$, this gives a round time of $L / c_\text{fiber}$, so that the repetition rate is $c_\text{fiber} / L$.
The end-to-end entanglement distribution rate is obtained by multiplying the repetition rate with the success probability, given by
\begin{equation}
    R = \frac {c_\text{fiber}} L \text{Pr(repeater chain)}  =  \frac {c_\text{fiber}} L \Big( \frac 1 2 \Big)^N \Big[ 1  - \Big( 1 - \frac 1 2 e^{-L/L_\text{att}}\Big)^M \Big]^{N+1}.
\end{equation}

\section{Proof of Equivalence}
\label{app:proof_of_equivalence}

In this appendix, we prove the equivalence between the link-based formulation and path-based formulation. In order to make this material self-contained, in \Cref{supp_not} we reintroduce some of the notations. Next, in \Cref{supp_sec:formulations}, we briefly re-describe both of the formulations. After that, in \Cref{supp_sec:path_to_link,supp_sec:link_to_path}, we show how to construct a feasible solution to the link-based formulation from the optimal solution to the path-based formulation and vice versa. By combining this result with the proof that the optimal objective values are equal, we conclude that the two formulations are equivalent. 
Here, we consider two ILP formulations to be equivalent if an optimal solution to one formulation can be used to obtain an optimal solution to the other and vice versa.

\begin{table}[H]
\centering
\begin{tabular}{|c |c | } 
 \hline
 $\mc Q$ & Set of all ordered pairs $(s,t)$ of end-nodes.  \\
  \hline
 $\mc P_q$ & Set of all possible paths for a pair $q \in \mc Q$. \\
  \hline
 $\mc P$ & Set of all possible paths between all of the pairs $q \in \mc Q$. \\
  \hline
 $\mc E_q$ & Set of all elementary links that can be used by a pair $q \in \mc Q$.\\
  \hline
 $\mc R$ & Set of potential repeater locations in the network.  \\ 
  \hline
 $L_{\text{max}}$ & Maximum length of an elementary link.\\
  \hline
 $N_{\text{max}}$ & Maximum number of repeaters in a path. \\
  \hline
 $K$ &
  The robustness parameter, which denotes the minimum number of quantum-repeater nodes or \\ & elementary links (it can be any combination) that need to break down before one of the other \\ & requirements can no longer be met.\\
  \hline
  $\mc K$ & Set of all integers from $1$ to $K$ (inclusive). \\
  \hline
 $D$ & The capacity parameter, which denotes the number of\\ & quantum-communication sessions that one quantum repeater can facilitate simultaneously. \\
 \hline
\end{tabular}
\caption{Overview of the sets and parameters that are relevant for the proof.}
\label{supp_not}
\end{table}

\subsection{Formulations}
\label{supp_sec:formulations}

Here we will restate both the path-based and link-based formulation for completeness. In the path-based formulation, we define the binary decision variables $x_p$ corresponding to a path $p \in \mathcal P = \cup_{(s,t) \in \mathcal Q} \mathcal P_{(s,t)}$, where $\mathcal P_{(s,t)}$ is the set of all possible paths from end node $s$ to end node $t$. A path itself is a sequence of elementary links reaching from $s$ to $t$ that does not contain any loops.
They have value $1$ when $p$ is considered part of the set of chosen paths, and $0$ otherwise.
Furthermore, we use the binary decision variables $y^p_u$ for all $u \in \mathcal R$. Note that we introduce the superscript $p$ here to more clearly distinguish between the two formulations, which differs from the main text.
The variable $y^p_u$ is $1$ if a quantum repeater is placed at potential repeater location $u$, and $0$ otherwise.
The exact formulation is given in \Cref{supp_pbf}. \\

\begin{table}[ht]
\boxalign{
\hspace*{-100pt}
\begin{alignat}{2}
    \hspace*{-100pt}
    \min \quad \sum_{u \in \mathcal R} y^p_u & \phantom{\leq} \label{supp_eq:pbf_objfun} \\
    \text{s.t. } \quad 
    L\Big((u,v)\Big)x_p &\leq L_\text{max} \qquad && \forall (u, v) \in p, p \in \mathcal P \label{supp_eq:pbf_max_length}\\
    |p| x_p &\leq N_\text{max} + 1 \qquad && \forall p \in \mathcal P \label{supp_eq:pbf_max_repeaters}\\
    \sum_{p \in \mathcal{P}_q} x_p & = K && \forall q \in \mathcal Q  \label{supp_eq:pbf_K}\\
    \sum_{p \in \mathcal{P}_q}r_{up}x_p & \leq 1 && \forall u \in \mathcal{R}, q \in \mathcal{Q} \label{supp_eq:pbf_disjoint}\\
    \sum_{p \in \mathcal P}r_{up}x_p & \leq Dy^p_u && \forall u \in \mathcal{R} \label{supp_eq:pbf_capacity}  \\
     x_p & \in \{0, 1\} && \forall p \in \mathcal P \label{supp_eq:pbf_x}\\
     y^p_{u} & \in \{0,1\} && \forall u \in \mathcal R \label{supp_eq:pbf_y} \\
     \mathrm{where} \quad
     r_{up} &= 
     \begin{cases}
     1 \qquad \text{if path $p$ uses $u$ as a quantum-repeater node}\\
     0 \qquad \text{otherwise}
     \end{cases} 
     \qquad
     &&\forall u \in \mathcal R, p \in \mathcal P \label{supp_eq:pbf_rup}
\end{alignat}
}
\caption{Path-based formulation.}
\label{supp_pbf}
\end{table}

On the other hand, in the link-based formulation we define the binary decision variable $x^{q,k}_{uv}$ for each pair of end nodes $q = (s,t) \in \mc Q$, every elementary link $(u,v) \in \mathcal E_q$ and $k \in \mc K = \{1,2, \dots K\}$.
These variables can be interpreted as indicating whether an elementary link $(u,v)$ is used in the $k^\text{th}$ path connecting the end nodes $s$ and $t$.
Furthermore, in the link-based formulation we use the variables $y^l_u$ to indicate whether node $u \in \mathcal R$ is used as a quantum-repeater node. The exact formulation is given in \Cref{supp_lbf}.

\begin{table}[ht]
\centering
\boxalign{
\begin{alignat}{2}
    \hspace*{-100pt}
    \min \quad \sum_{u \in \mathcal{R}} y^l_u & \phantom{\leq} \label{supp_eq:lbf_objfun} \\
    \mathrm{s.t.} \quad 
    \sum_{\substack{v \\ (u, v) \in \mathcal E _q}} x^{q,k}_{uv} - \sum_{\substack{v \\ (v, u) \in \mathcal E _q}}x^{q,k}_{vu} & =
    \begin{cases}
    1, \quad & \text{if } u = s \\
    -1, \quad & \text{if } u = t \\
    0, \quad & \text{if } u \in \mathcal{R}
    \end{cases}
    \qquad
    &&\forall u \in \mathcal R \cup\{s, t\}, q = (s, t) \in \mathcal{Q}, k \in \mc K \label{supp_eq:lbf_flowcon} \\
    L\Big( (u, v) \Big) x^{q,k}_{uv} & \leq L_{\mathrm{max}} &&  \forall (u,v) \in \mathcal E_q, q \in \mathcal{Q}, k \in \mc K \label{supp_eq:lbf_max_length} \\
    \sum_{(u,v) \in \mathcal E _q}  x^{q,k}_{uv} & \leq N_{\mathrm{max}} +1 && \forall q \in \mathcal{Q}, k \in \mc K \label{supp_eq:lbf_max_repeaters} \\
    \sum_{\substack{v \\ (u, v) \in \mathcal E_q}}\sum_{k \in \mc K} x_{uv}^{q,k} & \leq 1 && \forall u \in \mathcal{R}, q \in \mathcal{Q} \label{supp_eq:lbf_disjoint_repeaters} \\
    \sum_{k \in \mc K} x_{st}^{q,k} & \leq 1 && \forall q \in \mathcal{Q} \label{supp_eq:lbf_max_one_direct_path} \\
    \sum_{q \in\mathcal{Q}} \sum_{\substack{v \\ (u, v) \in \mathcal E _q}} \sum_{k \in \mc K} x^{q,k}_{uv} & \leq D y^l_u && \forall u \in \mathcal{R} \label{supp_eq:lbf_capacity} \\
    x^{q,k}_{uv} & \in \{0, 1\} && \forall (u,v) \in \mathcal{E}_q, q \in \mathcal{Q}, k \in \mc K \label{supp_eq:xdmon_lbf_k} \\
    y^l_u & \in \{0, 1\} && \forall u \in \mathcal{R} \label{supp_eq:ydom_lbf}
\end{alignat}
}
\caption{Link-based formulation.}
\label{supp_lbf}
\end{table}

\subsection{From the Path-Based Formulation to the Link-Based Formulation}
\label{supp_sec:path_to_link}

In this section we will construct a solution to the link-based formulation from the optimal solution to the path-based formulation. We then proceed by proving that this newly constructed solution is indeed a feasible solution to the link-based formulation, i.e. that it satisfies all constraints. \\

From the optimal solution to the path-based formulation we can use the values of the variables $x_p$ and $y_u^p$ to assign values to our new binary decision variables $\tl x_{uv}^{q,k}$ and $\tl y_u^l$, which presumably give a solution to the link-based formulation, using \Cref{alg:path_to_link}.

\smallskip
\begin{algorithm}[H]
	\SetAlgoLined
	For all $q\in \mc Q$, $k \in \mc K$ and $(u,v) \in \mc E_q$, set $\tl x^{q,k}_{uv} = 0$ and for all $u \in \mathcal{R}$, set $\tl y_u^l = y_u^p$\; 
	\For{$q \in \mc Q$}{
	    $k = 1$\;
		\For{$p \in \mc P_q$}{
		    \If{$x_p = 1$}{
		        \For{$(u,v) \in p$}{
		        $\tl x^{q,k}_{uv} = 1$\;
		        }
		    $k = k + 1$\;
		    }
		}
	}
	\caption{Methodology for assigning the values of $\tl x^{q,k}_{uv}$ and $\tl y_u^l$.}
	\label{alg:path_to_link}
\end{algorithm}
\smallskip

From \Cref{alg:path_to_link}, we can see that for each pair $q \in \mc Q$ and value of $k \in \mc K$, we select a single, unique path $p$ for which $x_p = 1$ 
(because there are exactly $K$ paths with $x_p=1$ in $\mathcal P_q$, due to Constraints \eqref{supp_eq:pbf_K})
and use its elementary links to assign the corresponding $\tl x_{uv}^{q,k}$ variables to have value $1$. Let us denote this path by $p^{q,k}$ for ease of notation in the remainder of this section. \\

Next, we will prove that the newly constructed solution is indeed feasible, by showing that all of the $\tl x^{q,k}_{uv}$ and $\tl y_u^l$ variables satisfy all constraints of the link-based formulation.

\begin{proposition}
    The variables $\tl x^{q,k}_{uv}$ that we obtain with \Cref{alg:path_to_link} satisfy Constraints \eqref{supp_eq:lbf_flowcon} of the link-based formulation, i.e.
    \begin{equation}
        \sum_{\substack{v \\ (u, v) \in \mathcal E_q}} \tl x^{q,k}_{uv} - \sum_{\substack{v \\ (v, u) \in \mathcal E _q}}\tl x^{q,k}_{vu}  =
    \begin{cases}
    1, \quad & \text{\upshape if } u = s \\
    -1, \quad & \text{\upshape if } u = t \\
    0, \quad & \text{\upshape if } u \in \mathcal{R}
    \end{cases}
    \quad \forall u \in \mc R \cup \{s, t\}, q = (s, t) \in \mc Q, k \in \mc K.
    \end{equation}
\end{proposition}

\begin{proof}
    Consider the path $p^{q,k}$ for specific values of $q = (s,t) \in \mathcal Q$ and $k \in \mathcal K$, for which $x_{p^{q,k}}=1.$.
    This path starts at $s$ and ends at $t$ by construction, so there is exactly one node $v \in \mc R \cup \{t\}$ for which $\tl x_{sv}^{q,k} = 1$ and hence $\sum_{v: (s, v) \in \mathcal{E}_q} \tl x_{sv}^{q,k} = 1$. Additionally, since the set $\mathcal{E}_q$ only contains outgoing edges from $s$, $\sum_{v: (v, s) \in \mathcal{E}_q} \tl x_{vs}^{q,k}$ is an empty sum and thus, trivially, $\sum_{v: (v, s) \in \mathcal{E}_q} \tl x_{vs}^{q,k} = 0$. This implies that
    \begin{equation}
        \sum_{\substack{v \\ (s, v) \in \mathcal E_q}} \tl x^{q,k}_{sv} - \sum_{\substack{v \\ (v, s) \in \mathcal E _q}}\tl x^{q,k}_{vs} = 1.
        \label{sup_eq:flowcon_source}
    \end{equation}
    In a similar fashion, 
    it must hold that there exists exactly one node $v \in R \cup \{s\}$ such that $\tl x_{vt}^{q,k} = 1$ and hence $\sum_{v: (v, t) \in \mathcal{E}_q}\tl x_{vt}^{q,k} = 1$. Furthermore, since the set $\mc E_q$ also only contains the incoming edges to $t$, it follows trivially that $\sum_{v: (t, v) \in \mc E_q} \tl x_{tv}^{q,k} = 0$, which implies that
    \begin{equation}
        \sum_{\substack{v \\ (t, v) \in \mathcal E_q}} \tl x^{q,k}_{tv} - \sum_{\substack{v \\ (v, t) \in \mathcal E _q}}\tl x^{q,k}_{vt} = -1.
        \label{sup_eq:flowcon_sink}
    \end{equation}
    If 
    $p^{q,k}$
    visits any other node $u \in \mc R$, it holds that there is exactly one incoming edge and one outgoing edge from this node, since a path cannot start nor end here and neither can a path contain loops by definition. 
    In other words, there must be exactly one node $v \in \mathcal{R} \cup \{s\}$ for which $\tl x_{vu}^{q,k} = 1$ and one node $v' \in \mathcal{R} \cup \{t\}$ (where $v \neq v'$) for which $\tl x_{uv'}^{q,k} = 1$. This results in
    
    \begin{equation}
        \sum_{\substack{v \\ (u, v) \in \mathcal E_q}} \tl x^{q,k}_{uv} - \sum_{\substack{v \\ (v, u) \in \mathcal E _q}}\tl x^{q,k}_{vu} = 0 \qquad \forall u \in \mc R.
        \label{sup_eq:flowcon_repnode}
    \end{equation}
    Finally, if a node $u\in \mathcal R$ is not visited by 
    $p^{q,k}$,
    there are no incoming and no outgoing edges. 
    Thus, $\tl x_{vu}^{q,k} = 0$ for all $v \in \mathcal R \cup \{s\}$, and $\tl x_{uv'}^{q,k} = 0$ for all $v' \in \mathcal R \cup \{t\}$.
    As a result, Equation \eqref{sup_eq:flowcon_repnode} is also satisfied in this case.
    The combination of \eqref{sup_eq:flowcon_source}, \eqref{sup_eq:flowcon_sink} and \eqref{sup_eq:flowcon_repnode} for all $q \in \mathcal Q$ and $k \in \mathcal K$ concludes the proof.
\end{proof}

\begin{proposition}
    The variables $\tl x^{q,k}_{uv}$ that we obtain with \Cref{alg:path_to_link} satisfy Constraints \eqref{supp_eq:lbf_max_length} of the link-based formulation, i.e.
    \begin{equation}
        L\big( (u, v) \big) \tl x^{q,k}_{uv} \leq L_{\mathrm{max}} \qquad \forall (u,v) \in \mathcal E_q, q \in \mathcal{Q}, k \in \mc K.
    \end{equation}
\end{proposition}

\begin{proof}
    Consider the path $p^{q,k}$ for specific values of $q = (s,t) \in \mathcal Q$ and $k \in \mathcal K$.
    Because $x_{p^{q,k}} = 1$ by definition, it follows from Constraints \eqref{supp_eq:pbf_max_length} that any elementary link $(u,v) \in p^{q,k}$ has length $L((u,v)) \leq L_\text{max}$.
    Since $\tl x_{uv}^{q,k} = 1$ only if $(u,v)$ is in $p^{q,k}$, it follows that every elementary link for which $\tl x_{uv}^{q,k} = 1$ must have a length smaller than or equal to $L_\text{max}$.
    Because this argument can be made for any $q\in \mathcal Q$ and any $k\in\mathcal K$, Constraints \eqref{supp_eq:lbf_max_length} are satisfied whenever $\tl x_{uv}^{q,k} = 1$.
    Furthermore, they are trivially satisfied when $\tl x_{uv}^{q,k} = 0$.
\end{proof}

\begin{proposition}
    The variables $\tl x^{q,k}_{uv}$ that we obtain with \Cref{alg:path_to_link} satisfy Constraints \eqref{supp_eq:lbf_max_repeaters} of the link-based formulation, i.e.
    \begin{equation}
        \sum_{(u,v) \in \mathcal E _q} \tl x^{q,k}_{uv} \leq N_{\mathrm{max}} +1 \qquad \forall q \in \mathcal{Q}, k \in \mc K.
    \end{equation}
\end{proposition}

\begin{proof}
    Consider the path $p^{q,k}$ for specific values of $q = (s,t) \in \mathcal Q$ and $k \in \mathcal K$, for which $x_{p^{q,k}} = 1$.
    According to \Cref{alg:path_to_link}, $\tl x_{uv}^{q,k}=1$ only if $(u,v) \in p^{q,k}$.
    Hence, the number of $(u,v) \in \mathcal E_q$ for which $\tl x_{uv}^{q,k} = 1$ is the same as the number of elementary links in $p^{q,k}$, which is $|p^{q,k}|$. 
    Thus,
    
    \begin{equation}
        |p^{q,k}| = |p^{q,k}| x_{p^{q,k}} = \sum_{(u, v) \in \mc E_q} \tl x_{uv}^{q,k},
        \label{supp_eq:number_of_elementary_links}
    \end{equation}
    where the first equality holds because $x_{p^{q,k}} = 1$ by definition of $p^{q,k}$.
    According to Constraints \eqref{supp_eq:pbf_max_repeaters},
    \begin{equation}
        |p| x_p \leq N_\text{max} + 1 \qquad \forall p \in \mathcal P.
    \end{equation}
    which implies
    \begin{equation}
        |p^{q,k}| x_{p^{q,k}} \leq N_\text{max} + 1 \qquad \forall q \in \mathcal Q, k \in \mathcal K. \label{sup_eq:proof_Nmax}
    \end{equation}
    If we substitute \eqref{supp_eq:number_of_elementary_links} into \eqref{sup_eq:proof_Nmax} we directly get that
    \begin{equation}
        \sum_{(u, v) \in \mc E_q} \tl x_{uv}^{q,k} \leq N_\text{max} + 1 \qquad \forall q \in \mathcal Q, k \in \mathcal K.
    \end{equation}

\end{proof}

\begin{proposition}
    The variables $\tl x^{q,k}_{uv}$ that we obtain with \Cref{alg:path_to_link} satisfy Constraints \eqref{supp_eq:lbf_disjoint_repeaters} and \eqref{supp_eq:lbf_max_one_direct_path} of the link-based formulation, i.e.
    \begin{equation}
        \sum_{\substack{v \\ (u, v) \in \mathcal E_q}}\sum_{k \in \mc K} \tl x_{uv}^{q,k} \leq 1 \qquad \forall u \in \mathcal{R}, q \in \mathcal{Q}
    \end{equation}
    and
    \begin{equation}
        \sum_{k \in \mc K} \tl x_{st}^{q,k} \leq 1 \qquad \forall q \in \mathcal{Q}.
    \end{equation}
\end{proposition}

\begin{proof}
    Consider the path $p^{q,k}$ for specific values of $q = (s,t) \in \mathcal Q$ and $k \in \mathcal K$, for which $x_{p^{q,k}} = 1$.
    According to Equation \eqref{supp_eq:pbf_rup}, the parameter $r_{up^{q,k}} = 1$ if the potential repeater location $u \in \mc R$ is in $p^{q,k}$.
    If this is the case, the path $p^{q,k}$ contains exactly one outgoing elementary link at $u$, and thus $\sum_{v : (u, v) \in \mathcal E_q} \tl x^{q,k}_{uv} = 1$.
    Otherwise, there are no outgoing elementary links, and $\sum_{v : (u, v) \in \mathcal E_q} \tl x^{q,k}_{uv} = 0$.
    Therefore,
    \begin{equation}
        r_{up^{q,k}} = \sum_{\substack{v \\ (u, v) \in \mathcal E_q}} \tl x^{q,k}_{uv}.
        \label{supp_eq:pbf_to_lbf_rup}
    \end{equation}
    Furthermore, we note that
    \begin{equation}
        \sum_{p \in \mathcal P_q} r_{up} x_p = \sum_{k \in \mathcal K} r_{up^{q,k}},
        \label{supp_eq:pbf_to_lbf_pqk_sum_over_k}
    \end{equation}
    since the paths $p \in \mathcal P_q$ for which $x_p=1$ are exactly the paths that were labeled $p^{q,k}$ for some $k \in \mathcal K$.
    Combining Equations \eqref{supp_eq:pbf_to_lbf_rup} and \eqref{supp_eq:pbf_to_lbf_pqk_sum_over_k}, and repeating this argument for every $u \in \mc R$ and $q \in \mc Q$ then gives
    \begin{equation}
        \sum_{p \in \mc P_q} r_{up}x_p = \sum_{\substack{v \\ (u, v) \in \mathcal E_q}} \sum_{k \in \mc K} \tl x^{q,k}_{uv} \quad \forall u \in \mc R, q \in \mc Q. \label{sup_eq:proof_disjoint_to_lbf}
    \end{equation}
    We can then substitute \eqref{sup_eq:proof_disjoint_to_lbf} directly into \eqref{supp_eq:pbf_disjoint}, in order to get
    \begin{equation}
        \sum_{\substack{v \\ (u, v) \in \mathcal E_q}} \sum_{k \in \mc K} \tl x^{q,k}_{uv} \leq 1 \quad \forall u \in \mc R, q \in \mc Q,
    \end{equation}
    which are exactly Constraints \eqref{supp_eq:lbf_disjoint_repeaters}. \\
    
    Additionally, the set $\mc P_q$ for $q = (s,t) \in \mathcal Q$ contains the path consisting of the direct elementary link $(s,t)$ exactly once.
    Furthermore, no other paths containing $(s,t)$ can exist because there exist no incoming edges to $s$ nor outgoing edges from $t$ in $\mc E_q$.
    If the path $p^{q,k'}$ is this direct path, then $\tl x_{st}^{q,k'} = 1$.
    But because $p^{q,k}$ cannot contain $(s,t)$ for $k \neq k'$, it holds that  $\tl x_{st}^{q,k} = 0$ for $k \neq k'$.
    Therefore,
        \begin{equation}
        \sum_{k \in \mc K} \tl x_{st}^{q,k}  = \tl x_{st}^{q,k'} + \sum_{\substack{k \in \mc K\\k \neq k'}} \tl x_{st}^{q,k} = 1.
    \end{equation}
    On the other hand, if there is no $k'$ such that $p^{q,k'}$ is the direct path,
    \begin{equation}
        \sum_{k \in \mc K} \tl x_{st}^{q,k}  = 0.
    \end{equation}
    In any case, the sum evaluates to smaller than or equal to 1.
    Since this argument holds for any $q \in \mathcal Q$, we can conclude that
    
    \begin{equation}
        \sum_{k \in \mc K} \tl x_{st}^{q,k} \leq 1 \qquad \forall q \in \mc Q. \label{sup_eq:proof_stpath}
    \end{equation}
    
\end{proof}

\begin{proposition}
    The variables $\tl x^{q,k}_{uv}$ and $\tl y_u^l$ that we obtain with \Cref{alg:path_to_link} satisfy Constraints \eqref{supp_eq:lbf_capacity} of the link-based formulation, i.e.
    \begin{equation}
        \sum_{q \in\mathcal{Q}} \sum_{\substack{v \\ (u, v) \in \mathcal E _q}} \sum_{k \in \mc K} \tl x^{q,k}_{uv} \leq D \tl y^l_u \qquad \forall u \in \mathcal{R}.
    \end{equation}
\end{proposition}

\begin{proof}
    If we sum over all $q \in \mc Q$ on both sides of \eqref{sup_eq:proof_disjoint_to_lbf}, we get
    
    \begin{align}
        \sum_{q \in \mc Q}\sum_{p \in \mc P_q} r_{up}x_p = \sum_{p \in \mc P}r_{up}x_p = \sum_{q \in Q}\sum_{\substack{v \\ (u, v) \in \mathcal E_q}} \sum_{k \in \mc K} \tl x^{q,k}_{uv} \qquad \forall u \in \mc R. \label{sup_eq:proof_capacity_to_lbf}
    \end{align}
    Substituting \eqref{sup_eq:proof_capacity_to_lbf} into \eqref{supp_eq:pbf_capacity} then results in
    \begin{equation}
        \sum_{q \in Q}\sum_{\substack{v \\ (u, v) \in \mathcal E_q}} \sum_{k \in \mc K} \tl x^{q,k}_{uv} \leq D y_u^p \qquad \forall u \in \mc R.
    \end{equation}
    Finally, since we assign all the values of $\tl y_u^l$ to have the same values as $y_u^p$ for all $u \in \mc R$ in \Cref{alg:path_to_link}, we conclude that
    \begin{equation}
        \sum_{q \in Q}\sum_{\substack{v \\ (u, v) \in \mathcal E_q}} \sum_{k \in \mc K} \tl x^{q,k}_{uv} \leq D \tl y_u^l \qquad \forall u \in \mc R.
    \end{equation}
\end{proof}

This concludes the proof that the variables we obtain from the optimal solution to the path-based formulation using \Cref{alg:path_to_link} provide a feasible solution to the link-based formulation. One important statement we make about the objective value of this newly constructed feasible solution is that

\begin{equation}
        \sum_{u\in \mc R} y^p_u = \sum_{u \in \mc R} \tilde{y}^l_u \geq \sum_{u \in \mc R} y^l_u \label{sup_eq:obj_val_feaslink},
\end{equation}

where $\sum_{u \in \mc R} y^l_u$ represents the optimal objective value of the link-based formulation. In other words, the objective value of the newly constructed feasible solution will always be greater than or equal to the objective value of the optimal solution, since we are solving a minimization problem.

\subsection{From the Link-Based Formulation to the Path-Based Formulation}
\label{supp_sec:link_to_path}

In this section we will go the other way around and construct a solution to the path-based formulation from the optimal solution to the link-based formulation using the path extraction algorithm, outlined in \Cref{alg:path_extraction}. First, we show that the application of this algorithm indeed leads to valid paths, after which we proceed by proving that this newly constructed solution is a feasible solution to the path-based formulation. \\

From the optimal solution to the link-based formulation, we can use the values of the variables $x_{uv}^{q,k}$ and $y_u^l$ to assign the values to our new binary decision variables $\tl x_p$ and $\tl y_u^p$, which presumably give a solution to the path-based formulation. We do this by setting $\tl y_u^p = y_u^l$ for all $u \in \mc R$, $\tl x_p = 1$ for all $p \in \mc P^*$ and $\tl x_p = 0$ for all $p \in \mc P \setminus \mc P^*$, where the set $\mc P^*$ is obtained 
from the path extraction algorithm. 
It follows from Proposition \ref{prop:valid_paths} that this can always be done.
Note that in the path extraction algorithm, to every $q \in \mathcal Q$ and $k \in \mathcal K$ there is a single path associated.
We label this path $p^{q,k}$ for the remainder of this section.
It follows that $x_{p^{q,k}} = 1$ for all $q \in \mathcal Q$ and $k \in \mathcal K$. \\

\begin{proposition}
\label{prop:valid_paths}
    \Cref{alg:path_extraction} is always successful.
    That is, it is always able to construct the set $\mathcal P^*$ such that $\mathcal P^* \subseteq \mathcal P$.
\end{proposition}
\begin{proof}
    We will prove this proposition by proving that \Cref{alg:link_to_path_sub} can always successfully construct a sequence $\bar p^{q,k}$ and moreover that this sequence forms a valid path (i.e. $\bar p^{q,k} \in \mathcal P$).
    Because \Cref{alg:path_extraction} is nothing but the repeated application of \Cref{alg:link_to_path_sub} (with $\mathcal P^* = \cup_{q,k}p^{q,k}$), it then follows that this proposition holds.
    
    \smallskip
    \begin{algorithm}[H]
	\SetAlgoLined
        $u_0 = s$\;
        $n = 0$ \;
         \While{$u_n \neq t$}
         {
            Find the unique node $v \in \mc R \cup \{t\}$ for which $x_{u_nv}^{q,k} = 1$\;
            $n = n + 1$\;
            $u_n = v$\;
	    }
	    $\bar p^{q,k} = \Big ( (u_0, u_1), (u_1, u_2), \dots, (u_{n_1}, u_n) \Big)$\;
	\caption{Path extraction sub-algorithm.}
	\label{alg:link_to_path_sub}
    \end{algorithm}
    \smallskip
    
    The first time when the algorithm enters the while loop, it has to find the single node $v$ for which $x^{q,k}_{sv} = 1$.
    To prove that there exists exactly one such node, we consider Constraints \eqref{supp_eq:lbf_flowcon} for $u = s$.
    Because there is no incoming elementary link at $s$, i.e. there is no $v'$ such that $(v', s) \in \mathcal E_q$, the second summation is empty and the equation reduces to
    \begin{equation}
        \sum_{\substack{v \\ (s,v) \in \mathcal E_q}} x_{sv}^{q,k} = 1.
    \end{equation}
    Since the variables $x_{uv}^{q,k}$ are binary, this implies that there is exactly one $v$ such that $x_{sv}^{qk}=1$. \\
    
    Now, we assume that $u = u_i$ such that $u_i \neq s$ and $u_i \neq t$ and show there is a unique $v$ such that $x_{u_i v}^{q,k} = 1$.
    First, we combine Constraints \eqref{supp_eq:lbf_disjoint_repeaters} with the fact that
    \begin{equation}
        \sum_{\substack{v \\ (u_i, v)\in \mathcal E_q}}x_{u_i v}^{q,k} \leq \sum_{\substack{v \\ (u_i, v) \in \mathcal E_q}}\sum_{k \in \mc K} x_{u_i v}^{q,k},
    \end{equation}
    to find that
    \begin{equation}
        \sum_{\substack{v \\ (u_i, v) \in \mathcal E_q}}x_{u_i v}^{q,k} \leq 1.
        \label{supp_eq:sum_x<1}
    \end{equation}
    From Constraints \eqref{supp_eq:lbf_flowcon} with $u_i \in \mathcal R$ we find that
    \begin{equation}
        \sum_{\substack{v \\ (u_i, v) \in \mathcal E_q}}x_{u_iv}^{q,k}  = \sum_{\substack{v \\ (v, u_i) \in \mathcal E_q}}x_{v u_i}^{q,k}. \qquad \forall q \in Q, k \in \mc K.
    \end{equation}
    We know that the left-hand side of this equation is upper bounded by $1$ because of Equation \eqref{supp_eq:sum_x<1}.
    Furthermore, because node $u_i$ was selected by \Cref{alg:link_to_path_sub} (when entering the while loop for $u=u_{i-1}$), we know that $x_{u_{i-1}u_i}^{q,k} = 1$. 
    This implies that the right-hand side is at least one.
    Therefore, both sides must be equal to one.
    Because the variables are binary, the equality of the right-hand side to $1$ implies there is exactly one $v$ such that $x_{u_i v}^{q,k}$ = 1.\\
    
    This procedure only concludes if there is an $n$ such that $u_n = t$.
    This must be the case, as there is only a finite number of nodes in $\mathcal N$, and two nodes $u_k$, $u_l$ cannot be the same unless $k = l$.
    To see that this last property holds, assume for the moment that there are a $k$ and $l>k$ such that $u_k = u_l$.
    In that case, by virtue of how \Cref{alg:link_to_path_sub} works, it must be the case that $x^{q,k}_{u_{k-1} u_k} = x^{q,k}_{u_{l-1}u_k} = 1$.
    Since we concluded earlier that there can only be one $v$ such that $x_{vu}^{q,k}=1$, this implies that $u_{l-1} = u_{k-1}$.
    Then, the above argument can be repeated to find $u_{l-2} = u_{k-2}$.
    This can be continued until we find that $u_{l-k} = u_{k-k} = u_0 = s$. 
    Because $l-k > 0$, $u_{l-k}$ can only be in $\bar p^{q,k}$ if $x_{u_{l-k-1} u_{l-k}}^{q,k} = x_{u_{l-k-1} s}^{q,k} = 1$.
    However, this variable is not defined, because there is no elementary link $(v, s) \in \mathcal E_q$ for any $v \in \mathcal N$.
    We have thus reached a contradiction, and we can conclude that $u_k \neq u_l$ as long as $k \neq l$ and thus \Cref{alg:link_to_path_sub} must eventually terminate. \\
    
    At this point we can conclude that \Cref{alg:link_to_path_sub} creates the sequence
    \begin{equation}
        \bar p^{q,k} = \Big ( (u_0, u_1), (u_1, u_2), (u_2, u_3), \dots, (u_{n-1}, u_n) \Big )
    \end{equation}
    for some integer $n$, where $u_0 = s$ and $u_n=t$.
    Clearly, this is a sequence of adjacent elementary links which connect the end node $s$ to the end node $t$.
    Furthermore, since we concluded that $u_k \neq u_l$ for $k \neq l$, there are no loops, and thus $\bar p^{q,k}$ is in fact a path $\bar p^{q,k} = p^{q,k} \in \mathcal P.$
    Since $p^{q,k} \in \mathcal P$, the variable $\tilde x_{p^{q,k}}$ is well-defined and can be set to $1$.
    
\end{proof}

Next, we address the fact that optimal solutions to the link-based formulation can contain chosen elementary links that form loops, which do not contribute to satisfying constraints, but also do not violate them.
To avoid the construction of ineffective elementary links,
we define the variables $\bar x^{q,k}_{uv}$ for all $q \in \mathcal Q$, $k \in \mathcal K$ and $(u,v) \in \mathcal E_q$.
We set $\bar x^{q,k}_{uv} = 1$ if $(u,v) \in p$ for some $p \in \mathcal P^*$ and $\bar x^{q,k}_{uv} = 0$ otherwise for all $q \in \mathcal Q$ and $k \in \mathcal K$, where $\mc P^*$ is the output of \Cref{alg:path_extraction} when applied to the variables $x_{uv}^{q,k}$ (which are part of an optimal solution to the link-based formulation). In other words, $\bar x_{uv}^{q,k}$ represent a choice of elementary links that corresponds to the optimal solution, but with all links that are not in any of the paths $p \in \mathcal P^*$ removed.

\begin{proposition} \label{prop:bar_x_valid}
    The variables $\bar x_{uv}^{q,k}$ and $y_u^l$ form an optimal solution to the the link-based formulation.
\end{proposition}

\begin{proof}
    If the variables form a feasible solution, they also form an optimal solution, since the variables $y_u^l$ are defined to be part of an optimal solution and the objective function is independent of the values of $\bar x_{uv}^{q,k}$.\\
    
    The variables form a feasible solution if they satisfy all constraints in \Cref{supp_lbf} (but with $x_{uv}^{q,k}$ substituted by $\bar x_{uv}^{q,k}$ everywhere).
    It is easily verified that Constraints (\ref{supp_eq:lbf_max_length} - \ref{supp_eq:lbf_capacity}) are satisfied.
    Each of these set an upper bound on sums over (linear functions of) $\bar x_{uv}^{q,k}$ variables.
    Since $\bar x_{uv}^{q,k}$ is either equal to $x_{uv}^{q,k}$ or set to $0$, it always holds that $\bar x_{uv}^{q,k} \leq x_{uv}^{q,k}$.
    Thus, replacing $x_{uv}^{q,k}$ variables by $\bar x_{uv}^{q,k}$ variables can only decrease the summations.
    Since the $x_{uv}^{q,k}$ variables satisfy all constraints by assumption (and the bounds are unaltered), we can conclude that all of these constraints are also satisfied by the $\bar x_{uv}^{q,k}$ variables.\\
    
    To show that Constraints \eqref{supp_eq:lbf_flowcon} are satisfied as well, consider the path $p^{q,k}$ for some specific $q = (s,t) \in \mathcal Q$ and $k \in \mathcal K$.
    By definition, $\bar x_{uv}^{q,k} = 1$ if and only if $(u,v) \in p^{q,k}$.
    By virtue of Proposition \ref{prop:valid_paths}, we know that $p^{q,k}$ is a valid path between $s$ and $t$, i.e.
    \begin{equation}
        p^{q,k} = \Big( (u_0, u_1), (u_1, u_2), \dots , (u_{n_1}, u_n) \Big)
    \end{equation}
    for some integer $n$, where $u_0 = s$, $u_n = t$, $u_i \in \mathcal R$ for $0 < i < n$, and $u_i \neq u_j$ for $i \neq j$.\\
    
    Consider Constraints \eqref{supp_eq:lbf_flowcon} for $u = s$.
    $x_{sv}^{q,k} = 1$ holds if and only if $v = u_1$, and thus $\sum_{v: (s,v) \in \mathcal E_q} \bar x_{sv}^{q,k} = \bar x_{su_1}^{q,k} = 1$.
    Furthermore, since there are no incoming elementary links at $s$, i.e. $(v, s) \notin \mathcal E_q$ for all $v \in \mathcal R \cup \{s, t\}$, $\sum_{v:(v,s) \in \mathcal E_q} \bar x_{vs}^{q,k} = 0$ trivially.
    Therefore,
    \begin{equation}
        \sum_{\substack{v \\ (s,v) \in \mathcal E_q}}\bar x_{sv}^{q,k} - \sum_{\substack{v \\ (v,s) \in \mathcal E_q}}\bar x_{vs}^{q,k} = 1
    \end{equation}
    and therefore, Constraints \eqref{supp_eq:lbf_flowcon} hold for $u=s$.\\
    
    When $u=t$, we can make a similar argument: 
    $\bar x_{vt}^{q,k} = 1$ holds if and only if $v = u_{n-1}$, and thus $\sum_{v: (v,t) \in \mathcal E_q} x_{vt}^{q,k}=1$.
    Furthermore, there are no elementary links leaving $t$, i.e. there is no $v \in \mathcal R \cup \{s, t\}$ such that $(t, v) \in \mathcal E_q$.
    Therefore, $\sum_{v: (t, v) \in \mathcal E_q} \bar x_{tv}^{q,k} =0$ trivially.
    Thus we can conclude Constraints \eqref{supp_eq:lbf_flowcon} hold for $u = t$, i.e.
    \begin{equation}
        \sum_{\substack{v \\ (t,v) \in \mathcal E_q}}\bar x_{tv}^{q,k} - \sum_{\substack{v \\ (v,t) \in \mathcal E_q}}\bar x_{vt}^{q,k} = -1.
    \end{equation}
    
    When $u \in \mathcal R$, it can be the case that there is an $i$ such that $u = u_i$ for $1 < i < n$.
    Then, $\bar x_{u_i v}^{q,k} = 1$ holds if and only if $v = u_{i+1}$, and therefore $\sum_{v: (u_i,v) \in \mathcal E_q} \bar x_{u_iv}^{q,k} = \bar x_{u_i u_{i+1}}^{q,k}= 1$.
    Furthermore, $\bar x_{v u_i}^{q,k} = 1$ holds if and only if $v = u_{i-1}$, and therefore $\sum_{v: (v, u_i) \in \mathcal E_q} \bar x_{v u_i}^{q,k} = \bar x_{u_{i-1} u_{i}}^{q,k}= 1$.
    On the other hand, if there is no $i$ such that $u=u_i$, $u$ is not on the path.
    It then holds by definition that $\bar x_{uv}^{q,k} = \bar x_{vu}^{q,k} = 0$ for all $v \in \mathcal R \cup \{s, t\}$.
    In both cases, it follows directly that
    \begin{equation}
        \sum_{\substack{v \\ (u,v) \in \mathcal E_q}}\bar x_{uv}^{q,k} - \sum_{\substack{v \\ (v,u) \in \mathcal E_q}}\bar x_{vu}^{q,k} = 0.
    \end{equation}
    Therefore, Constraints \eqref{supp_eq:lbf_flowcon} also holds for $u \in \mathcal R$.
    We conclude that all constraints hold, and thus $\bar x_{uv}^{q,k}$ and $y_u^l$ together form an optimal solution to the link-based formulation.
\end{proof}

Note that, by definition, for all $q \in \mathcal Q$ and $k \in \mathcal K$, there are no elementary links $(u,v) \in \mathcal E_q$  such that $\bar x_{uv}^{q,k} = 1$ which are not in the path $p^{q,k}$.
This property, together with Proposition \ref{prop:bar_x_valid}, makes it easier to show that all constraints of the path-based formulation are satisfied for the variables $\tl x_p$.
\\

\begin{proposition}
    The variables $\tl x_p$ that we obtain satisfy Constraints \eqref{supp_eq:pbf_max_length} of the path-based formulation, i.e.
    \begin{equation}
        L\big((u, v)\big)\tl x_p \leq L_\text{max} \qquad \qquad \forall (u, v) \in p, p \in \mathcal P.
    \end{equation}
\end{proposition}

\begin{proof}
    Every path $p$ for which $\tilde x_p=1$ is equal to $p^{q,k}$ for some $q \in \mathcal Q$ and $k \in \mathcal K$.
    For every elementary link $(u,v)$ that makes up $p^{q,k}$ it holds that $\bar x_{uv}^{q,k}=1$.
    Therefore, due to Constraints \eqref{supp_eq:lbf_max_length}, each elementary link $(u,v) \in p^{q,k}$ must have length $L((u,v)) = L((u,v)) \tilde x_{p^{q,k}} \leq L_\text{max}$.
    For all paths $p$ which are not equal to $p^{q,k}$ for some $q \in \mathcal Q$ and $k \in \mathcal K$, $\tilde x_p=0$ and $L((u,v)) \tilde x_p = 0 \leq L_\text{max}$.
    Therefore, $L((u,v)) \tilde x_p \leq L_\text{max}$ holds for any $(u,v) \in p$ for any path $p \in \mathcal P$.
\end{proof}

\begin{proposition}
    The variables $\tl x_p$ that we obtain satisfy Constraints \eqref{supp_eq:pbf_max_repeaters} of the path-based formulation, i.e.
    \begin{equation}
        |p| \tl x_p \leq N_\text{max} + 1 \qquad \qquad \forall p \in \mathcal P.
    \end{equation}
\end{proposition}

\begin{proof}
    Consider the path $p^{q,k}$ for specific values of $q = (s,t) \in \mathcal Q$ and $k \in \mathcal K$.
    By definition, the number of elementary links in $p^{q,k}$ is the same as the number of elementary links for which $\bar x_{uv}^{q,k} = 1$.
    Thus,
    \begin{equation}
        \sum_{(u, v) \in \mc E_q} \bar x_{uv}^{q,k} = |p^{q,k}| = |p^{q,k}| \tl x_{p^{q,k}}.
        \label{sup_eq:proof_Nmaxpbf}
    \end{equation}
    Substituting \eqref{sup_eq:proof_Nmaxpbf} into \eqref{supp_eq:lbf_max_repeaters} then gives
    \begin{equation}
        |p^{q,k}| \tl x_{p^{q,k}} \leq N_\text{max} + 1.
    \end{equation}
    Furthermore, for any path $p \in \mathcal P$ which is not equal to $p^{q,k}$ for some $q \in \mathcal Q$ and $k \in \mathcal K$, it holds that $\tilde x_p=0$ and thus $|p| \tilde x_p = 0 \leq N_\text{max} + 1$.
    Therefore, we can conclude that $|p|\tilde x_p \leq N_\text{max} + 1$ for any $p \in \mathcal P$.
\end{proof}

\begin{proposition}
    The variables $\tl x_p$ that we obtain satisfy Constraints \eqref{supp_eq:pbf_K} of the path-based formulation, i.e.
    \begin{equation}
        \sum_{p \in \mathcal{P}_q} \tl x_p = K \qquad \forall q \in \mathcal Q.
    \end{equation}
\end{proposition}

\begin{proof}
    Consider the set $\mathcal P_q$ for a specific value of $q \in \mathcal Q$.
    Every path $p \in \mathcal P_q$ has $\tilde x_p = 1$ if it is equal to $p^{q,k}$ for some $k \in \mathcal K$ and $\tilde x_p = 0$ otherwise.
    Therefore,
    \begin{equation}
        \sum_{p \in \mathcal P_q} \tilde x_p = \sum_{k \in \mathcal K} \tilde x_{p^{q,k}} =\sum_{k \in \mathcal K} 1 =  K \qquad \forall q \in \mc Q.
    \end{equation}
\end{proof}

\begin{proposition}
    The variables $\tl x_p$ that we obtain satisfy Constraints \eqref{supp_eq:pbf_disjoint} of the path-based formulation, i.e.
    \begin{equation}
        \sum_{p \in \mathcal{P}_q}r_{up}\tl x_p \leq 1 \qquad \forall u \in \mathcal{R}, q \in \mathcal{Q}.
    \end{equation}
\end{proposition}

\begin{proof}
     Consider the path $p^{q,k}$ for specific values of $q = (s,t) \in \mathcal Q$ and $k \in \mathcal K$.
     The parameter $r_{up^{q,k}}$ is defined in Equation \eqref{supp_eq:pbf_rup}, and takes the value $1$ if the path $p^{q,k}$ passes the potential repeater location $u \in \mc R$ and $0$ otherwise.
     Since $p^{q,k}$ can only pass $u$ if there is an outgoing elementary link from $u$ in $p^{q,k}$, and since the elementary links in $p^{q,k}$ are exactly those elementary links $(u,v)$ for which $\bar x^{q,k}_{uv} = 1$,
     $r_{up^{q,k}} = 1$ if and only if $\bar x^{q,k}_{uv} = 1$ for some $v$ such that $(u,v) \in \mathcal E_q$.
     Therefore,
     \begin{equation}
        r_{up^{q,k}} =   \sum_{\substack{v \\ (u, v) \in \mathcal E_q}} \bar x_{uv}^{q,k}.
     \end{equation}
     Since $x_{p^{q,k}} = 1$, this directly implies that
      \begin{equation}
         r_{up^{q,k}} \tilde x_{p^{q,k}} = \sum_{\substack{v \\ (u, v) \in \mathcal E_q}} \bar x_{uv}^{q,k}.
        \label{supp_eq:lbf_to_pbf_rup}
     \end{equation}
     Furthermore, since for every path $p \in \mathcal P_q$, $\tl x_p = 1$ if $p$ is equal to $p^{q,k}$ for some $k \in \mathcal K$ and $\tl x_p = 0$ otherwise,
     \begin{equation}
         \sum_{p \in \mathcal P_q}r_{up} \tilde x_p = \sum_{k \in \mathcal K} r_{up^{q,k}} \tilde x_{p^{q,k}}.
         \label{supp_eq:lbf_to_pbf_pqk_sum_over_k}
     \end{equation}
     Combining Equation \eqref{supp_eq:lbf_to_pbf_rup} and Equation \eqref{supp_eq:lbf_to_pbf_pqk_sum_over_k} then gives
     \begin{equation}
         \sum_{p \in \mathcal P_q}r_{up} \tilde x_p = \sum_{\substack{v \\ (u, v) \in \mathcal E_q}} \sum_{k \in \mathcal K} \bar x_{uv}^{q,k}.
         \label{sup_eq:proof_rup}
     \end{equation}
    Substituting \eqref{sup_eq:proof_rup} into \eqref{supp_eq:lbf_disjoint_repeaters} gives
    \begin{equation}
        \sum_{p \in \mc P_q} r_{up} \tl x_p \leq 1.
    \end{equation}
    Repeating this argument for every every $u \in \mathcal{R}$ and $q \in \mathcal{Q}$ results in Constraints \eqref{supp_eq:pbf_disjoint}.
    
\end{proof}

\begin{proposition}
    The variables $\tl x_p$ and $\tl y_u^p$ that we obtain satisfy Constraints \eqref{supp_eq:pbf_capacity} of the path-based formulation, i.e.
    \begin{equation}
        \sum_{p \in \mathcal P}r_{up}\tl x_p \leq D\tl y^p_u \qquad \forall u \in \mathcal{R}.
    \end{equation}
\end{proposition}

\begin{proof}
    If we substitute \eqref{sup_eq:proof_rup} into \eqref{supp_eq:lbf_capacity} (with every $x_{uv}^{q,k}$ replaced by $\bar x_{uv}^{q,k}$), we get that
    \begin{align}
        \sum_{q \in \mc Q}\sum_{p \in \mc P_q}r_{up}\tl x_p & = \sum_{p \in \mc P}r_{up}\tl x_p \leq Dy_u^l \qquad \forall u \in \mc R. \label{sup_eq:proof_capacity}
    \end{align}
    Additionally, we assign all the values of $\tl y_u^p$ to have the same value as $y_u^l$ for all $u \in \mc R$, such that we can replace $y_u^l$ with $\tl y_u^p$ in \eqref{sup_eq:proof_capacity}  to get
    \begin{equation}
        \sum_{p \in \mc P}r_{up}\tl x_p \leq D \tl y_u^p \qquad \forall u \in \mc R.
    \end{equation}
\end{proof}

This concludes the proof that the variables we obtain from the optimal solution to the link-based formulation provide a feasible solution to the path-based formulation.
Another important observation we can make about the objective value of this newly constructed feasible solution is that 

\begin{equation}
        \sum_{u\in \mc R} y^l_u = \sum_{u \in \mc R} \tilde{y}^p_u \geq \sum_{u \in \mc R} y^p_u \label{sup_eq:obj_val_feaspath}.
\end{equation}
If we now combine \eqref{sup_eq:obj_val_feaslink} and \eqref{sup_eq:obj_val_feaspath}, we reach the conclusion that
\begin{equation}
    \sum_{u\in \mc R} y^l_u = \sum_{u\in \mc R} y^p_u,
\end{equation}
i.e. the optimal objective values of the path-based formulation and link-based formulation are equal. This implies that the feasible solution to the link-based formulation we obtain from the path-based formulation and vice versa are actually optimal solutions. We thus conclude that the two formulations are equivalent. 